\documentclass[aps,prx,twocolumn,footinbib,floatfix,superscriptaddress]{revtex4-2}
\usepackage{graphicx}
\usepackage{indentfirst}
\usepackage{braket} 
\usepackage{float}
\usepackage{amsmath}
\usepackage{physics}
\usepackage{amssymb}
\usepackage{CJK}
\usepackage{esint}
\usepackage{color}
\usepackage[T1]{fontenc}
\usepackage{subfigure}
\usepackage{amsfonts}
\usepackage{footmisc}
\usepackage{scrextend}
\usepackage{multirow}
\usepackage[outdir=./]{epstopdf}

\usepackage[english]{babel}
\usepackage{url}
\usepackage{comment}
\usepackage{array}
\usepackage{subfiles}
\usepackage{tikz}
\usetikzlibrary{shapes,arrows}
\usepackage{mathrsfs}
\usepackage[mathscr]{eucal}

\usepackage[hyperfootnotes=false]{hyperref}
\usepackage{cleveref}
\definecolor{darkblue}{rgb}{0,0,0.5}
\hypersetup{
colorlinks=true,
linkcolor=black,
filecolor=blue,
citecolor=darkblue,  
urlcolor=black,
}

\newtheorem{theorem}{Theorem}

\newtheorem{lemma}[theorem]{Lemma}

\newenvironment{proof}[1][Proof]{\noindent\textbf{#1.} }{\ \rule{0.5em}{0.5em}}

\def\be{\begin{equation}}
\def\ee{\end{equation}}
\def\ba{\begin{eqnarray}}
\def\ea{\end{eqnarray}}
\def\bal{\begin{equation}\begin{aligned}}
\def\eal{\end{aligned}\end{equation}}

\def\d{^\dagger}
\def\bp{\begin{pmatrix}}
\def\ep{\end{pmatrix}}

\usepackage{bm}

\newcommand{\calC}{{\cal C}}
\newcommand{\calD}{{\cal D}}
\newcommand{\calE}{{\cal E}}
\newcommand{\calF}{{\cal F}}
\newcommand{\calI}{{\cal I}}
\newcommand{\calL}{{\cal L}}
\newcommand{\calM}{{\cal M}}
\newcommand{\calN}{{\cal N}} 
\newcommand{\calT}{{\cal T}}

\newcommand{\calH}{{\cal H}}

\newcommand{\calK}{{\cal K}}

\newcommand{\calS}{{\cal S}}
\newcommand{\calU}{{\cal U}}

\newcommand{\p}{^\prime}
\newcommand{\1}{^{(1)}}

\newcommand{\state}[1]{\ketbra{#1}{#1}}

\newcommand{\QZ}[1]{{{\textcolor{black}{#1}}}}

\usepackage[normalem]{ulem}

\begin{document}
\title{Entanglement-assisted capacity regions and protocol designs for quantum multiple-access channels}
\author{Haowei Shi}
\affiliation{
James C. Wyant College of Optical Sciences, University of Arizona, Tucson, Arizona 85721, USA
}

\author{Min-Hsiu Hsieh}
\affiliation{
Hon Hai Research Institute, Taipei, 114, Taiwan
}

\author{Saikat Guha}
\affiliation{
James C. Wyant College of Optical Sciences, University of Arizona, Tucson, Arizona 85721, USA
}

\author{Zheshen Zhang}
\affiliation{
James C. Wyant College of Optical Sciences, University of Arizona, Tucson, Arizona 85721, USA
}
\affiliation{
Department of Materials Science and Engineering, University of Arizona, Tucson, Arizona 85721, USA
}

\author{Quntao Zhuang}
\email{zhuangquntao@email.arizona.edu}
\affiliation{
James C. Wyant College of Optical Sciences, University of Arizona, Tucson, Arizona 85721, USA
}
\affiliation{
Department of Electrical and Computer Engineering, University of Arizona, Tucson, Arizona 85721, USA
}

\begin{abstract}
We solve the entanglement-assisted (EA) classical capacity region of quantum multiple-access channels with an arbitrary number of senders. As an example, we consider the bosonic thermal-loss multiple-access channel and solve the one-shot capacity region enabled by an entanglement source composed of sender-receiver pairwise two-mode squeezed vacuum states. The EA capacity region is strictly larger than the capacity region without entanglement-assistance. With two-mode squeezed vacuum states as the source and phase modulation as the encoding, we also design practical receiver protocols to realize the entanglement advantages. Four practical receiver designs, based on optical parametric amplifiers, are given and analyzed. In the parameter region of a large noise background, the receivers can enable a simultaneous rate advantage of $82.0\%$ for each sender. Due to teleportation and superdense coding, our results for EA classical communication can be directly extended to EA quantum communication at half of the rates. Our work provides a unique and practical network communication scenario where entanglement can be beneficial.

\end{abstract}

\maketitle

\section{Introduction}
Communication channels model physical media for information transmission. In the case of a single-sender single-receiver channel, the Shannon capacity theorem~\cite{Shannon_1948,cover1999elements} concludes that a channel is essentially characterized by a single quantity---the channel capacity. As physical media obey quantum physics, the channel model eventually needs to incorporate quantum effects during the transmission, which has re-shaped our understanding of communication. To begin with, the Shannon capacity has been generalized to the Holevo-Schumacher-Westmoreland (HSW) classical capacity~\cite{hausladen1996classical,schumacher1997sending,holevo1998capacity}. Quantum effects such as entanglement have also enabled non-classical phenomena in communication, such as superadditivity~\cite{hastings2009superadditivity,smith2008quantum,zhu2017,zhu2018superadditivity,leditzky2018,fanizza2020quantum} and capacity-boost from entanglement-assistance (EA)~\cite{bennett1992,bennett1999entanglement,bennett2002entanglement,holevo02,shor2004classical,hsieh2008entanglement,zhuang2017additive,wilde2012quantum,wilde2012information,zhuang2020entanglement}. Moreover, reliable transmission of quantum information is possible, established by the Lloyd-Shor-Devetak quantum capacity theorem~\cite{quantum_capacity_Lloyd,quantum_capacity_Shor,quantum_capacity_Devetak}. Combining different types of information transmission, Refs.~\onlinecite{Wilde2012,Wilde2012_2} provide a capacity formula for the simultaneous trade-off of classical information (bits), quantum information (qubits) and quantum entanglement (ebits). 

Despite their exact evaluation being prevented by the superadditivity dilemma, capacities of single-sender single-receiver quantum channels are well-understood. In particular, the benefits of entanglement in boosting the classical communication rates have been known since the pioneering theory works~\cite{bennett1992,bennett2002entanglement,bennett1999entanglement,holevo02,hsieh2008entanglement} and recently experimentally demonstrated~\cite{hao2021entanglement} in a thermal-loss bosonic communication channel. The two-mode-squeezed-vacuum (TMSV) state is utilized as the entanglement source and functional quantum receivers are demonstrated, thanks to the practical protocol design in Ref.~\cite{shi2020practical}. Further development of receiver designs~\cite{guha2020infinite} and the application to covert communication~\cite{gagatsos2020covert} have also been considered.

However, supported by the Internet, real-life communication scenarios, such as online lectures and online conferences, often involve multiple senders and/or receivers. As a common paradigm being studied in the literature~\cite{laurenza2017general,notzel2020entanglement,leditzky2020playing,yard2008capacity,qi2018applications}, the multiple-access channel (MAC) concerns multiple senders and a single receiver. Communication over a MAC is no longer characterized by a single rate, but a rate region with a trade-off between multiple senders. With the development of a quantum network~\cite{kimble2008quantum,biamonte2019complex,wehner2018quantum,kozlowski2019towards,zhang2021entanglement}, quantum effects have also become relevant in such a communication scenario. In this regard, the classical capacity region of a quantum MAC was solved by Winter~\cite{winter2001capacity}, while the entanglement-assisted (EA) classical communication capacity region in the special case of a two-sender MAC was solved in Ref.~\cite{hsieh2008entanglement}. Although superadditivity in the capacity region has also been found in a MAC~\cite{czekaj2009purely,czekaj2011} and EA advantage in a classical MAC can be shown~\cite{leditzky2020playing}, it is unclear how much advantage entanglement can provide for a quantum MAC in a direct communication scenario.

In this work, we present a thorough study of EA classical communication over a quantum MAC with an arbitrary number of senders. On the fundamental information-theoretic side, we prove the general EA classical capacity theorem for an $s$-sender ($s\ge 2$) MAC, which has been conjectured in Ref.~\cite{hsieh2008entanglement} and yet not proven for the past decade. Next, we proceed to evaluate the EA rate region of the bosonic thermal-loss MAC, which models an optical or microwave communication scenario (see Fig.~\ref{fig:concept}), and find rigorous advantages from entanglement.  Finally, on the application layer, we propose practical protocols to realize the EA advantage in a bosonic thermal-loss MAC, and provide a variety of transmitter and receiver designs. Due to teleportation~\cite{bennett1993teleporting} and superdense coding~\cite{bennett1992}, our results for EA classical communication can be directly extended to EA quantum communication at half of the rates. As bosonic thermal-loss MACs model various real-world communication networks, our EA communication scenario is widely applicable to radio-frequency, deep-space~\cite{banaszek2019approaching}, and wireless communication scenarios~\cite{win2000ultra}.

\section{Results}

\begin{figure}[tbp]
    \centering
    \includegraphics[width=0.45\textwidth]{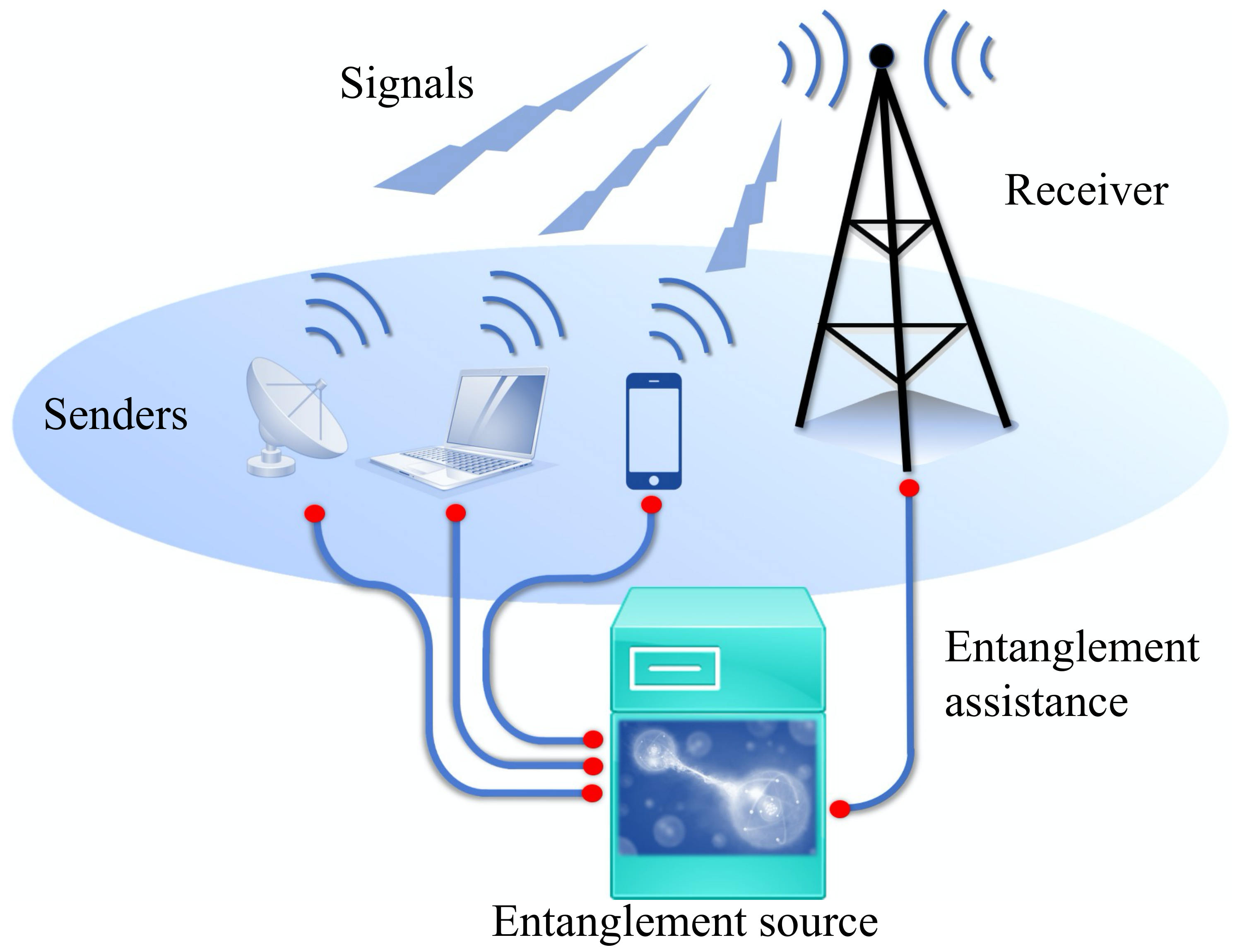}
    \caption{{\bf Conceptual schematic of EA classical communication over a MAC.} The entanglement source distributes entangled pairs to each sender and the receiver, potentially via a quantum network. The senders encode their own message on their share and send the signals to the receiver. The receiver decodes the messages of all senders by jointly measuring the received signal and the entanglement assistance locally stored.
    \label{fig:concept}
    }
\end{figure}

In a MAC, multiple senders individually communicate with a single receiver. As shown in Fig.~\ref{fig:concept}, besides a transmitter that sends encoded messages, each sender has access to entanglement pre-shared with the receiver, potentially through a ground-satellite and/or fiber-based quantum network. The receiver decodes all messages from the senders via a joint measurement on all received signals and the stored EA. Our first main result is an EA capacity theorem which quantifies the trade-off between the ultimate communication rates of different senders. The capacity formula has a form of conditional quantum mutual information, analogous to the classical formula~\cite{cover1999elements}. We then give an explicit example of a bosonic thermal-loss MAC and evaluate its rate-region with the common TMSV entanglement source. Comparing it with the case without EA~\cite{yen2005multiple}, we find great advantages enabled by entanglement; Moreover, when the sources of all senders have equal and low brightness, we numerically find that the TMSV source is optimal at a corner rate point. As a benchmark, we derive bounds on the capacity region and design practical protocols, based only on off-the-shelf quantum optical elements, which can achieve quantum advantages from entanglement in the near-term.


\subsection{ EA classical capacity theorem for MAC}

\begin{figure}[tbp]
    \centering
    \includegraphics[width=0.475\textwidth]{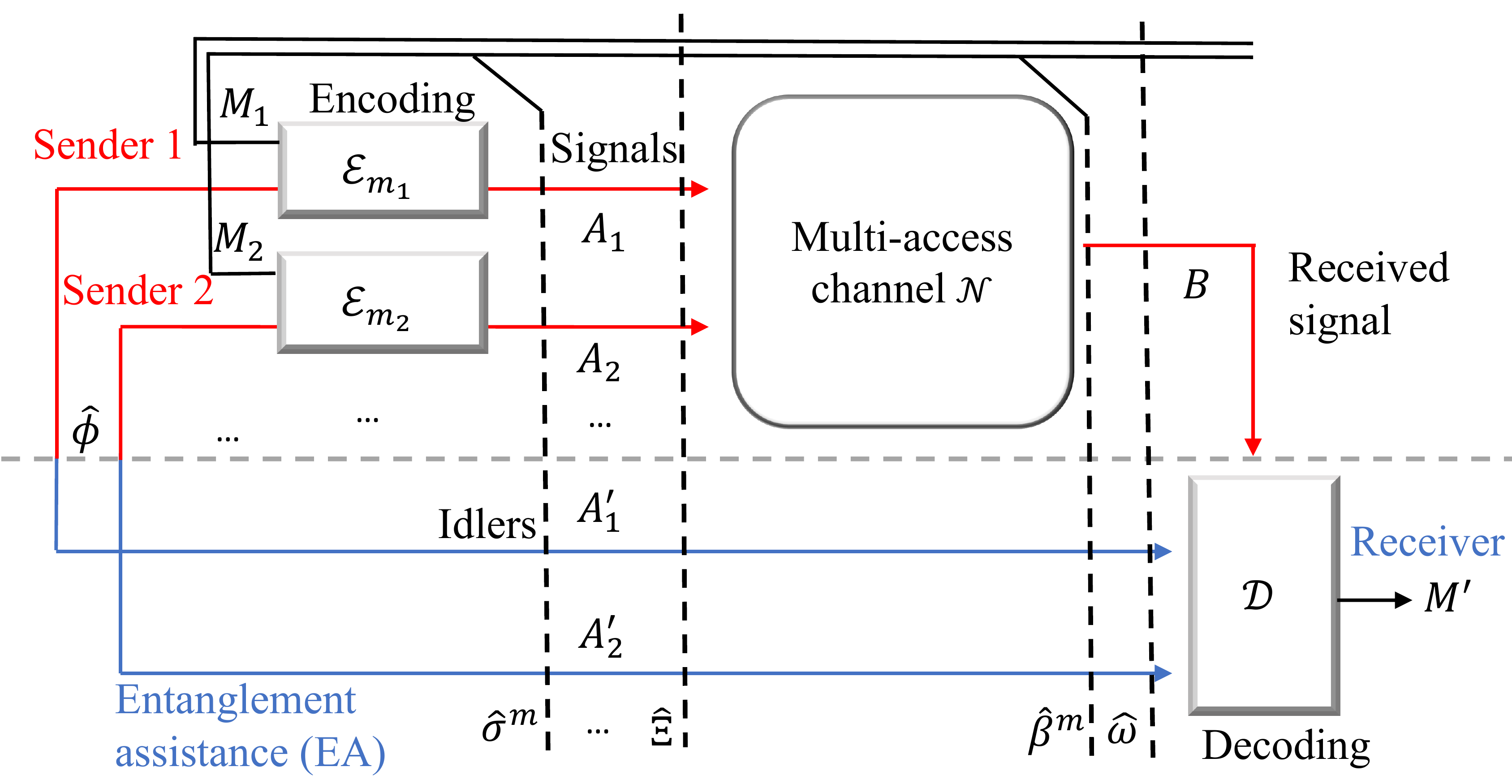}
    \caption{ {\bf Schematic of a general EA-MAC communication protocol.} The EA sources $\hat{\phi}$ of the $s$ senders are in a product state of Eq.~\eqref{EA_state}. The $s$ senders apply independent encoding modeled by quantum operations, i.e., sender $k$ applies $\calE_{m_k}$ on the signal state given the message $m_k$. Denoting the entire message as $m=m_1\cdots m_s$, the encoded signal-idler is then in a state $\hat{\sigma}^m$. The senders' encoded quantum systems $A=A_1\cdots A_s$ are sent through the MAC $\calN$, leading to the output system ${B}$. The receiver applies the quantum operation $\calD$ to decode the information from the joint state $\hat{\beta}^m$ of the output system ${B}$ and the pre-shared reference systems $A'=A_1'\cdots A_s'$. We define $M_k$ as the codeword space of each message $m_k$, $M$ as the overall codeword space of message $m$, and $M^\prime$ as the decoded codeword space. To facilitate the analysis, we denote the overall state $\hat{\Xi}$ (Eq.~\eqref{eq:Xi}) over systems $MAA^\prime$ right before the channel and the overall state $\hat{\omega}$  (Eq.~\eqref{eq:omega}) over systems $MBA^\prime$ right before the decoding.
    \label{fig:EAMACschematic_main}
    }
\end{figure}

\subsubsection{Multiple-access channels}
As depicted in Fig.~\ref{fig:EAMACschematic_main}, consider a MAC with $s$ senders, each sending a message $m_k$ ($1\le k \le s$) sampled from a message space $M_k$, therefore the overall message $m=m_1\cdots m_s$ is sampled from the message space $M=\otimes_{k=1}^s M_k$. To send each message $m_k$, the $k$th sender performs a quantum operation $\calE_{m_k}$ to produce a signal quantum system $A_k$. Following Ref.~\cite{hsieh2008entanglement}, we introduce EA in the above communication scenario---namely the receiver has a reference system $A_k^\prime$ (idler) pre-shared as the EA with the $k$th sender.

We consider the entanglement to be pairwise between each sender and the receiver such that the overall quantum state
\be 
\hat{\phi}_{AA^\prime}=\otimes_{k=1}^s \hat{\phi}_{A_k A_k^\prime}
\label{EA_state}
\ee 
is in a product form, where we have denoted $A=A_1\cdots A_s$ and $A^\prime=A_1^\prime\cdots A_s^\prime$ as the overall systems.

After the encoding, the composite system $A$ containing all of the quantum systems $\{A_k\}_{k=1}^s$ is input to the MAC $\calN_{A\to B}$, which outputs the quantum system $B$ for the receiver to decode the messages jointly with the EA $A^\prime$. 
For convenience, we define a quantum state after the channel but without the encoding,
$
\hat{\rho}_{BA^\prime}= \left[\calN_{A\to B}  \otimes \calI\right] \left(\hat{\phi}_{AA^\prime}\right),
$
where $\calI$ is the identity channel modeling the ideal storage of the idler system. The formal analyses of the quantum-state evolution can be found in Appendix~\ref{app:method}.

The performance metric of the above communication scenario is described by a vector of rates $(R_1,\cdots,R_s)$, where $R_k$ is the reliable communication rate between the $k$th sender and the receiver (see Appendix~\ref{app:method}, Section II of Ref.~\cite{winter2001capacity}, and Subsection III.A of Ref.~\cite{hsieh2008entanglement} for the formal definitions). These rates in general have non-trivial trade-offs with each other. In the case without EA, the capacity region is well-established by the pioneering work of Winter~\cite{winter2001capacity} (see Appendix~\ref{supp:note2}). 

To describe the rate region of the $s$-sender MAC, we will frequently divide the senders into two blocks, the block of interest indexed by a sequence $J$ and the complementary block $J^c$. For example, when $s=2$, we have four possible block divisions: $\{J=1, J^c=2\}$, $\{J=2,J^c=1\}$, and and two trivial cases $\{J=12, J^c=\varnothing\}$, $\{J=\varnothing, J^c=12\}$. 
Any $s$-fold quantity can be written as a composition of the two blocks, e.g., message space $M=M[J]M[J^c]$, with $M[J]=\otimes_{i\in J} M_i$, $M[J^c]=\otimes_{i\in J^c}M_i$; similarly the message as $m=m[J]m[J^c]$.

\subsubsection{Capacity theorem} To present our EA-MAC capacity theorem for the scenario in Fig.~\ref{fig:EAMACschematic_main}, we introduce some entropic quantities. For a quantum system $XYZ$ in a state $\hat{\alpha}$, we define the quantum mutual information between $X$ and $Y$ as 
\be 
I(X:Y)_{\hat{\alpha}}=S(X)_{\hat{\alpha}}+S(Y)_{\hat{\alpha}}-S(XY)_{\hat{\alpha}},
\nonumber
\ee 
where $S(X)_{\hat{\alpha}}=S(\hat{\alpha}_X)=-\tr \left(\hat{\alpha}_X\log_2 \hat{\alpha}_X\right)$ is the von Neumann entropy. Similarly, the quantum conditional mutual information between $X$ and $Z$ conditioned on $Y$
\be 
I(X;Z|Y)_{\hat{\alpha}}=S(XY)_{\hat{\alpha}}+S(YZ)_{\hat{\alpha}}-S(XYZ)_{\hat{\alpha}}-S(Y)_{\hat{\alpha}}.
\nonumber
\ee 
With the entropic quantities in hand, we can present our main theorem below (see Appendix~\ref{supp:note3} for a proof). 
\begin{theorem}[EA-MAC capacity]
\label{theorem: EA_MAC_main}
The entanglement-assisted classical communication capacity region over an $s$-sender MAC $\calN$ is given by the regularized union 
\be
\calC_{\rm E}(\calN)=\overline{\bigcup_{\ell=1}^\infty \frac{1}{\ell}  \calC_{\rm E}^{(1)}(\calN^{\otimes \ell})}
\label{eq:CE_regularization}
\ee 
where the ``one-shot'' capacity region $\calC_{\rm E}^{(1)}(\calN)$ is the convex hull of the union of ``one-shot, one-encoding'' regions
\be 
\calC_{\rm E}^{(1)}(\calN)={\rm Conv}\left[\bigcup_{\hat{\phi}} \tilde \calC_{\rm E}(\calN,\hat{\phi} )\right].
\label{eq:CE_union_state}
\ee 
The ``one-shot, one-encoding''  rate region $\tilde \calC_{\rm E}(\calN,\hat{\phi} )$ for the 2s-partite pure product state $\hat{\phi}_{AA^\prime}=\otimes_{k=1}^s\hat{\phi}_{A_kA_k^\prime}$ over $AA^\prime$, is the set of rates $(R_1,\cdots,R_s)$ satisfying the following $2^s$ inequalities
\bal
\sum_{k\in J}R_k&\leq I(A'[J];B|A'[J^c])_{\hat{\rho}}, \forall J,
\label{eq:CeMAC_qinfo}
\,\eal 
where the conditional quantum mutual information is evaluated over the output state $\hat{\rho}_{BA^\prime}=\calN_{A\to B}\otimes\calI(\hat{\phi}_{AA^\prime})$. 
\end{theorem}

Here we make some remarks about Theorem~\ref{theorem: EA_MAC_main}: First, if we only focus on the regularized capacity region $\calC_{\rm E}(\calN)$, then the convex hull in Eq.~\eqref{eq:CE_union_state} is not necessary, as one can simply include the time-sharing over different inputs among the infinite number of channel uses; However, if one wants to formulate the ``one-shot'' capacity region $\calC_{\rm E}^{(1)}(\calN)$, then the convex hull is necessary to include potential time-sharing between any codes. 
Second, the capacity formula in Ref.~\cite{hsieh2008entanglement} can be considered as a special case of our theorem, as the regularized case does not need the convex hull in Eq.~\eqref{eq:CE_union_state}; indeed, at the end of Ref.~\cite{hsieh2008entanglement} our theorem is stated as a conjecture.

\subsection{ EA capacity region for a bosonic MAC}

\begin{figure}[tbp]
    \centering
    \includegraphics[width=0.35\textwidth]{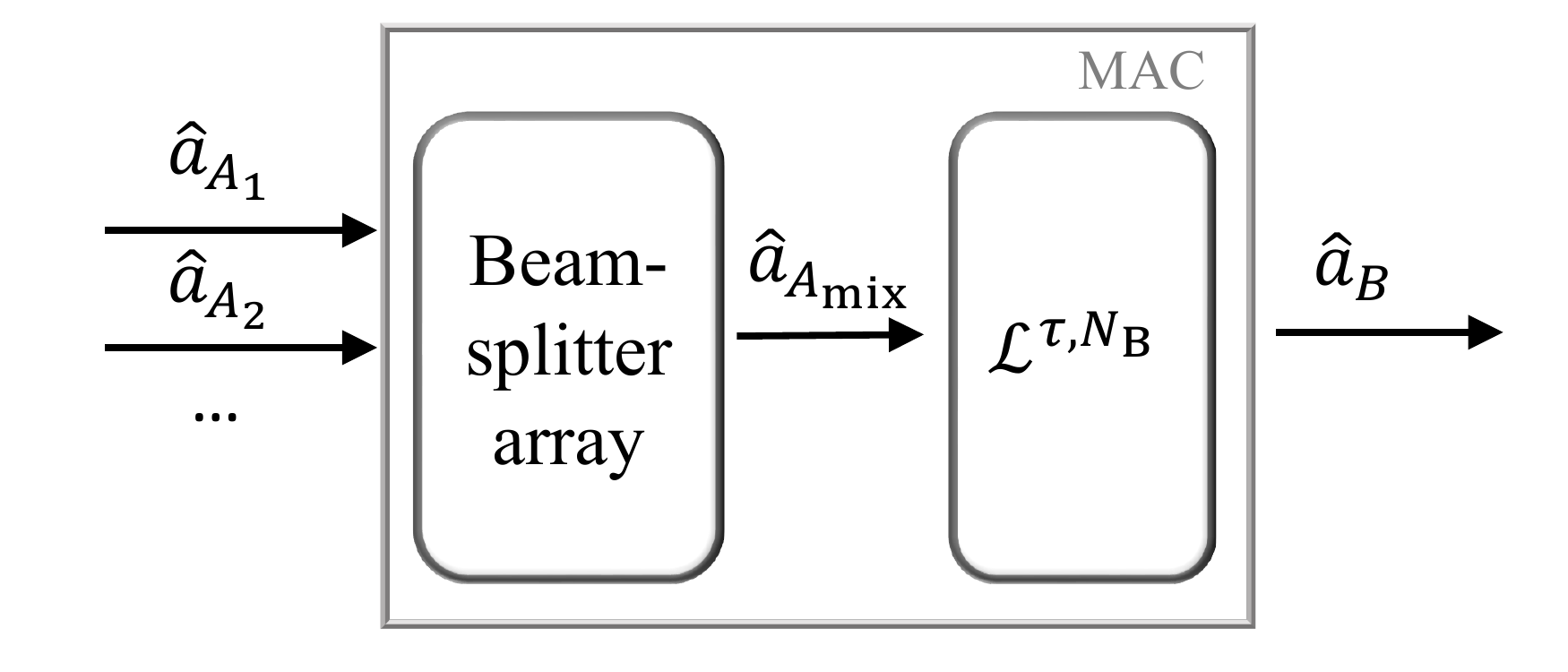}
    \caption{{\bf Schematic of the bosonic thermal-loss MAC.} The beam-splitter array models a linear scattering medium. The thermal-loss channel models the noisy transmission.
    }
    \label{fig:MACbosonic}
\end{figure}

\subsubsection{Bosonic thermal-loss multiple-access channels}
In an optical or microwave communication scenario, the relevant MAC is a bosonic thermal-loss MAC depicted in Fig.~\ref{fig:MACbosonic}. Upon the input modes $\hat{a}_{A_1}\cdots \hat{a}_{A_s}$ from the $s$ senders, the MAC $\calN$ first combines the modes through a beam-splitter array to produce a mixture mode 
\be
\hat{a}_{A_{\rm mix}}=\sum_{k=1}^s \sqrt{\eta_k}\hat{a}_{A_k},
\label{Amix}
\ee 
while all other ports of the beam-splitter array are discarded, here the weights $\{\eta_k\}$ are non-negative and normalized. Then the mixture mode goes through a bosonic thermal-loss channel $\calL^{\tau,N_{\rm B}}$ described by the operator transform
\be 
\hat{a}_B=\sqrt{\tau} \hat{a}_{A_{\rm mix}}+\sqrt{1-\tau} \hat{a}_E,
\label{thermal_loss}
\ee 
where $\hat{a}_E$ denotes the environment mode in a thermal state with a mean photon number $\expval{\hat{a}_E^\dagger \hat{a}_E}=N_{\rm B}/(1-\tau)$. This convention of fixing the mean photon number $N_{\rm B}$ of the thermal noise mixed into the output mode $\hat{a}_B$ is widely used, e.g., in quantum illumination~\cite{Tan2008,zhuang2017}.

In a bosonic MAC, the Hilbert space of the quantum systems is infinite-dimensional---an arbitrary number of photons can occupy a single mode due to the bosonic nature of light. To model a realistic communication scenario, we will consider an energy constraint on the mean photon number (brightness) of the signals modes
\be 
\expval{\hat{a}_{A_k}^\dagger \hat{a}_{A_k}}=N_{{\rm S},k}, 1\le k \le s,
\ee 
which is commonly adopted in bosonic communication~\cite{giovannetti2014ultimate,yen2005multiple,shi2020practical}. Note that in general the energy of different senders can be different.

Without EA, the capacity region of the above bosonic MAC has been considered in Ref.~\cite{yen2005multiple} for the two-user case. However, the generalization of the coherent-state rate region therein to the $s$-sender case is straightforward, leading to a rate region specified by the following $2^s$ inequalities,
\be
\sum_{i\in J} R_i\le C_{\rm coh}^J\equiv g\left(\sum_{i\in J} \tau\eta_i N_{{\rm S},i}+N_{\rm B}\right)-g\left(N_{\rm B}\right),
\label{C_J_region}
\ee
where $g(x)=(x+1)\log_2(x+1)-x\log_2(x)$ and $J$ can be chosen arbitrarily. Moreover, a squeezing-based encoding scheme is shown to be advantageous over the coherent-state encoding; however, regardless of the encoding, the rate region is always bounded by the following set of outer bounds
\be  
R_k\le g\left(\tau N_{{\rm S},k}+N_{\rm B}\right)-g\left(N_{\rm B}\right), \, 1\le k\le s,
\label{C_outer_bound}
\ee 
which are derived by assuming a super receiver that can reverse the beamsplitter array in the bosonic thermal-loss MAC. A second outer bound can be obtained from energetic considerations, which leads to the same form of Ineq.~\eqref{C_J_region} with $J$ being all users. As these outer bounds represent the upper limit of all encodings without EA, an EA rate region outside the rate region specified by the above outer bounds will demonstrate a strict advantage enabled by entanglement.

\subsubsection{EA outer bounds}As the exact evaluation of the EA capacity region for the bosonic MAC is challenging, we first focus on outer bounds to obtain some insights.
Similar to the case without EA, via reducing to the single-sender EA classical capacity, one can obtain outer bounds for the EA-MAC classical capacity region (See Appendix~\ref{app:method} for a proof).
Explicitly, we have
\begin{subequations}
\label{eq:EA_outer_bound}
\begin{align}
&R_k\leq C_{\rm E}\left( N_{{\rm S},k}; \calL^{\tau, N_{\rm B}}\right), 1\le k \le s,
\label{eq:EA_outer_bound_individual}
\\
&\sum_{k=1}^s R_k\leq C_{\rm E}\left(\sum_{k=1}^s \eta_k N_{{\rm S},k}; \calL^{\tau, N_{\rm B}}\right)
\label{eq:EA_outer_bound_sum}
\,,
\end{align}
\end{subequations}
where the explicit formula of the EA capacity $C_{\rm E}\left(N_{\rm S}; \calL^{\tau, N_{\rm B}}\right)$ over a bosonic thermal-loss channel $\calL^{\tau, N_{\rm B}}$, with the energy constraint $N_{\rm S}$, can be found in Eq.~\eqref{CE_full} of Appendix~\ref{app:method}. These outer bounds provide the upper limit of EA classical communication rates, and apply to arbitrary forms of entanglement source $\hat{\phi}$ and encoding $\{\calE_m\}$.

\subsubsection{Two-mode squeezed vacuum rate region} 
To obtain an explicit example of bosonic EA-MAC capacity region, we consider the entanglement source in Eq.~\eqref{EA_state} as a product of TMSV pairs, each with the wave-function 
\be 
\hat{\phi}_{A_k A_k^\prime}^{\rm TMSV}=\sum_{n_k=0}^\infty \sqrt{\frac{N_{{\rm S},k}^{n_k}}{(N_{{\rm S},k}+1)^{n_k+1}}} \ket{n_k}_{A_k}\ket{n_k}_{A_k^\prime},
\label{eq:state_TMSV}
\ee 
for $1\le k \le s$, where $\ket{n}$ is the number state defined by $\hat{a}^\dagger \hat{a}\ket{n}=n\ket{n}$.
In Ref.~\cite{shi2020practical}, it has been shown that the TMSV state is optimal for single-sender single-receiver EA classical communication, therefore we expect the TMSV source to be good in the MAC case. Although, due to the complexity from the plurality of the senders, the exact union over the states in Eq.~\eqref{eq:CE_union_state} for the EA-MAC classical capacity region is challenging to solve.

We evaluate the ``one-shot, one-encoding'' rate region $\tilde \calC_{\rm E}(\calN,\hat{\phi}^{\rm TMSV} )$ in Ineqs.~\eqref{eq:CeMAC_qinfo} for the TMSV source in Eq.~\eqref{eq:state_TMSV}. Although the evaluation of each Ineq.~\eqref{eq:CeMAC_qinfo} is efficient thanks to the Gaussian nature of the state, the number of such inequalities $2^s$ is exponential and therefore resource-consuming in practice. To showcase the capacity region, we choose $s=2,3$, which enable direct visualization as the rate region is two or three dimensional. In comparison, we also compute the classical coherent-state rate region in Ineq.~\eqref{C_J_region} and the classical outer bound, specified jointly by Ineq.~\eqref{C_outer_bound} and Ineq.~\eqref{C_J_region} with $J$ being all senders. Moreover, we can also compare $\tilde \calC_{\rm E}(\calN,\hat{\phi}^{\rm TMSV} )$ with the EA outer bound in Ineqs.~\eqref{eq:EA_outer_bound}.

Three representative setups of parameters are chosen as examples. To begin with, we consider an intermediate channel noise $N_{\rm B}=20$, identical to the case of microwave quantum illumination~\cite{Tan2008,guha2009}; Furthermore, a noisy channel with sufficiently large noise $N_{\rm B}=10^4$ is noteworthy as it provides a saturated EA advantage~\cite{shi2020practical}; Finally, the long wavelength infrared domain with relatively small noise $N_{\rm B}=0.1$ is a relatively uncharted territory for EA communication, nevertheless also relevant for practical application.

\begin{figure}[tbp]
    \centering
    \hspace{-.5cm}\includegraphics[width=0.425\textwidth]{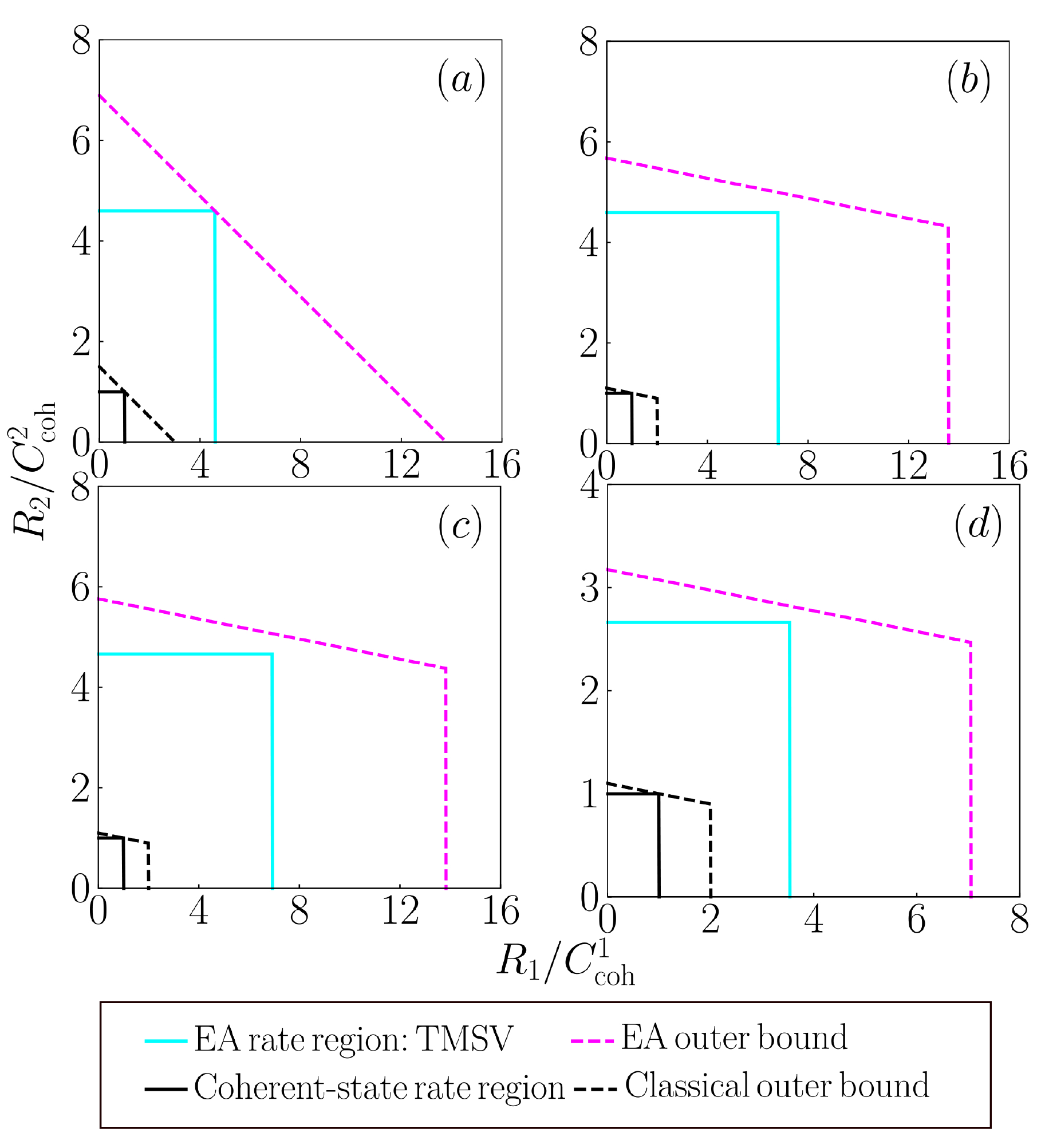}
    \caption{{\bf The symmetric two-sender rate region.} Rates are normalized by the coherent-state bound $C_{\rm coh}^1$, $C_{\rm coh}^2$ defined in Ineq.~\eqref{C_J_region}, evaluated in the scenario of (a) microwave domain, $\tau=0.01,N_{\rm B}=20$, $N_{{\rm S},1}=N_{{\rm S},2}=0.01$, $\eta_1=1/3$, $\eta_2=2/3$. (b) microwave domain, $\tau=0.01,N_{\rm B}=20$, $\eta_1=\eta_2=1/2$, $N_{{\rm S},1}=0.001, N_{{\rm S},2}=0.01$. (c) a noisy channel, $\tau=10^{-3},N_{\rm B}=10^4$, $\eta_1=\eta_2=1/2$, $N_{{\rm S},1}=0.001, N_{{\rm S},2}=0.01$.
    (d) long wavelength infrared domain $\eta_1=\eta_2=1/2,\tau=0.001,N_{\rm B}=0.1$, $N_{{\rm S},1}=0.001, N_{{\rm S},2}=0.01$.
    The EA rate region in Ineq.~\eqref{eq:CeMAC_qinfo} (cyan solid), evaluated on TMSV states, is bounded by the EA outer bound (magenta dashed) in Ineqs.~\eqref{eq:EA_outer_bound}; while the coherent-state rate region (black solid) given by Ineq.~\eqref{C_J_region} is bounded by the classical outer bound (black dashed).
    \label{fig:rateregion_Ce}
    }
\end{figure}

We begin with a two-sender case ($s=2$).
As shown in Fig.~\ref{fig:rateregion_Ce}, in all the parameter settings being considered, we can see strict advantages of the EA capacity region (cyan solid) over the classical outer bound (black dashed), which is higher than the coherent-state rate region (black solid). We find that the advantage is larger when the noise $N_{\rm B}$ is larger, comparing subplots (c) and (d). In particular, this advantage also holds when $N_{\rm S}\ll N_{\rm B} \ll1$, which can happen in the long wavelength infrared domain, as shown in subplot (d).

Comparing with the EA outer bound (magenta dashed), we see that in Fig.~\ref{fig:rateregion_Ce}(a) the TMSV rate region (cyan solid) touches the EA outer bound (magenta dashed) at a corner point when $R_2/C_{\rm coh}^2=R_1/C_{\rm coh}^1$ to the leading order. The gap is of the order of $10^{-5}$ relatively; therefore, at this point, the TMSV source is in fact optimal for the thermal-loss MAC being considered, for this symmetric case where the parameters $N_{{\rm S},k}\ll1$ are identical among the senders. Note this holds although the transmissivities of the senders $\eta_{k}$ are not equal. In other cases, when $N_{{\rm S},1}\neq N_{{\rm S},2}$, regardless of the values of $\eta_{k}$ being equal, a strict gap between the TMSV rate region and the EA outer bound exists. This does not conclude that the TMSV encoding is inferior, though, as the outer bound is likely to be loose.

\begin{figure}[tbp]
    \centering
    \includegraphics[width=0.4\textwidth]{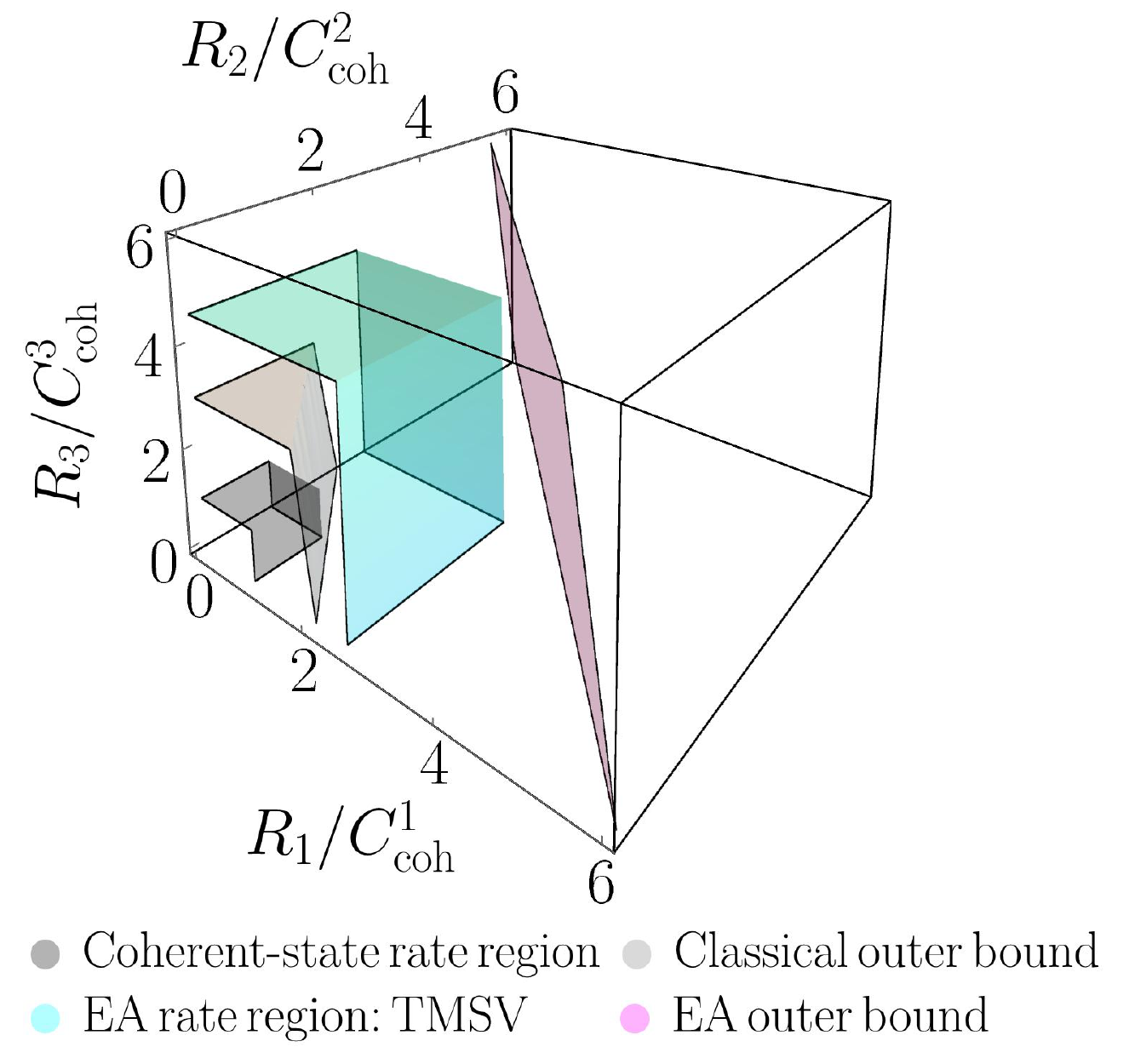}
    \caption{{\bf The asymmetric three-sender EA rate region.} The rates are normalized by the coherent state bound $C_{\rm coh}^1$, $C_{\rm coh}^2$ and $C_{\rm coh}^3$ defined in Ineq.~\eqref{C_J_region}, evaluated in the scenario of microwave domain $N_{{\rm S},1}=N_{{\rm S},2}=0.1, N_{{\rm S},3}=0.01,\tau=0.01,N_{\rm B}=20,\eta_1=\eta_2=\eta_3=1/3$.
    The EA rate region in Ineq.~\eqref{eq:CeMAC_qinfo} (cyan), evaluated on TMSV states, is bounded by the EA outer bound (magenta) in Ineqs.~\eqref{eq:EA_outer_bound}; while the coherent-state rate region (black) given by Ineq.~\eqref{C_J_region} is bounded by the classical outer bound (light gray).
    \label{fig:rate3Dregion_Ce}
    }
\end{figure}

Furthermore, we consider a three-sender asymmetric case ($s=3$), with unequal source brightness $N_{{\rm S},1}=N_{{\rm S},2}\neq N_{{\rm S},3}$. In Fig.~\ref{fig:rate3Dregion_Ce}, a gap emerges between the TMSV rate region (the region below the cyan surface) and the outer bound (the magenta surface), as we expected. An appreciable EA advantage remains as the EA capacity region is several times larger than the coherent state rate region (dark gray surface) and the classical outer bound (light gray surface).

\begin{figure}[tbp]
    \centering
    \includegraphics[width=0.475\textwidth]{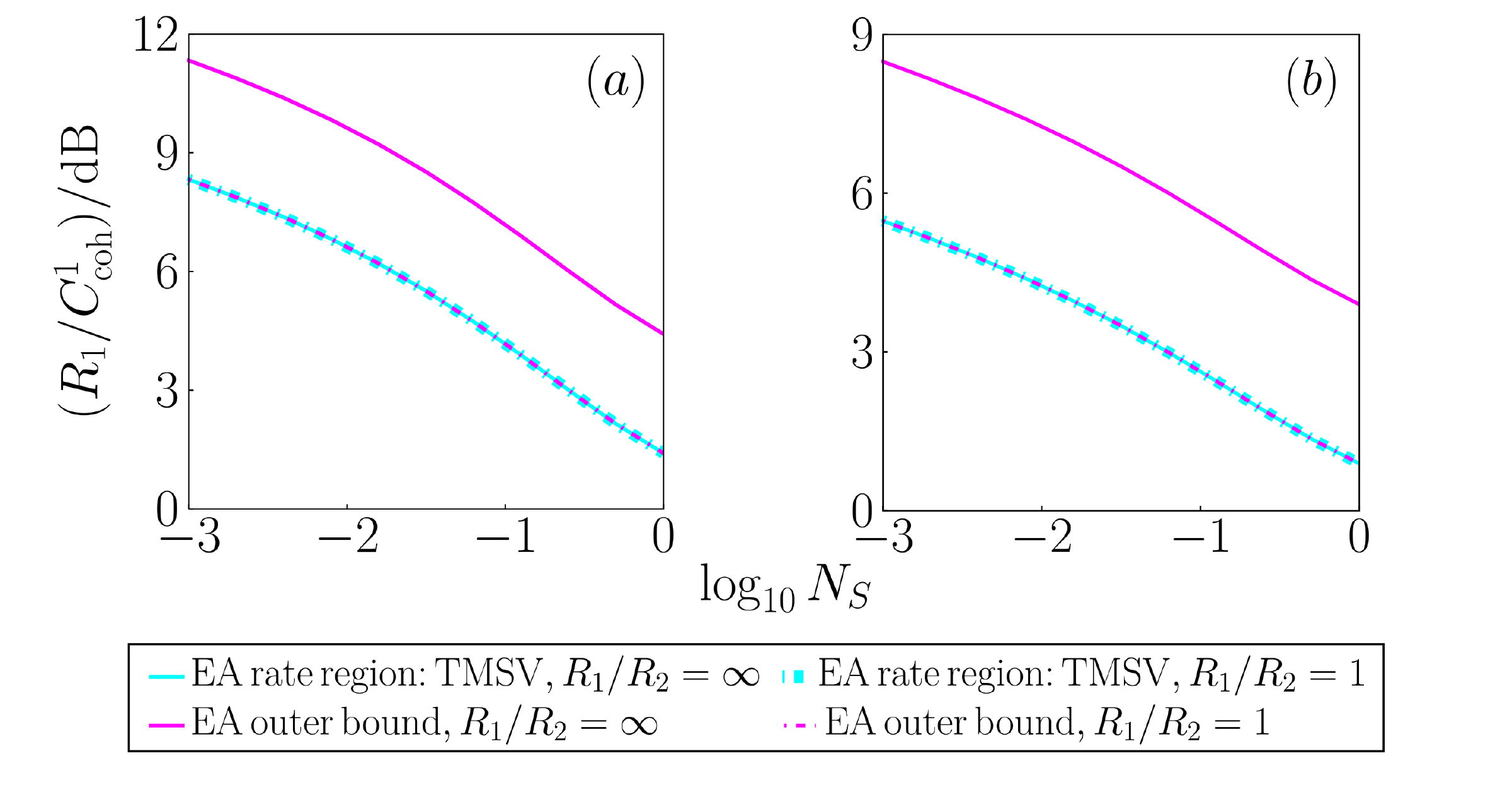}
    \caption{{\bf Rates versus signal brightness.} The dependence on source brightness $N_{{\rm S},1}=N_{{\rm S},2}=N_{\rm S}$ of the EA advantage of the EA rate regions for two-sender MAC communication under the scenario of (a) microwave domain $\eta=1/2,\tau=0.01,N_{\rm B}=20$; (b) long wavelength infrared domain $\eta_1=\eta_2=1/2,\tau=0.001,N_{\rm B}=0.1$. We plot $R_1$ for sender 1 under conditions $R_1/R_2=\infty$ (solid) and $R_1/R_2=1$ (dot-dashed). For TMSV, the two curves overlap. Note that $R_1/R_2=0,\infty$ are equivalent up to a swap due to the symmetry between the two senders; and for given $R_1/R_2$, $R_2/C^2_{\rm coh}=R_1/C^1_{\rm coh}$. We also compare the EA rate region of TMSV (cyan) with the EA outer bound (magenta). 
    \label{fig:Ns_Ce}
    }

\end{figure}

Now we further consider the scaling of the EA advantage observed above.
As shown in Fig.~\ref{fig:Ns_Ce}, the advantage of the EA capacity (magenta) relative to the case without EA also diverges with $\log(N_{\rm S})$, when the signal brightness $N_{\rm S}$ is small and the noise $N_{\rm B}$ is much larger than the signal brightness $N_{\rm S}$. Note that this advantage also holds for the case when $N_{\rm S}\ll N_{\rm B}<1$, as shown in Fig.~\ref{fig:Ns_Ce}(b). This logarithmic diverging EA advantage in MAC is similar to the single-sender single-receiver case studied in Ref.~\cite{shi2020practical}. Indeed, at the limit $\tau\ll1, N_{\rm S}\ll 1$, the relative ratio of the outer bound over the coherent-state rate
\be
\frac{C_{\rm E}\left( N_{{\rm S},k}; \calL^{\tau, N_{\rm B}}\right)}{C_{\rm coh}^k}\simeq \frac{\log (1/N_{{\rm S},k})}{\eta_k(1+N_{\rm B})\log(1+1/N_{\rm B})},
\label{eq:Ce/C}
\ee
is also logarithmic in $1/N_{{\rm S},k}$ when $N_{\rm B}$ is small.

\subsection{Protocol designs for the bosonic EA-MAC}
\noindent In this section, we design a practical protocol to realize EA classical communication over the bosonic thermal-loss MAC. The protocol consists of phase-modulation encoding on the TMSV entanglement source and structured receiver designs.

\subsubsection{Encoding and receiver designs}

\begin{figure*}[tbp]
\centering
\includegraphics[width=0.7\textwidth]{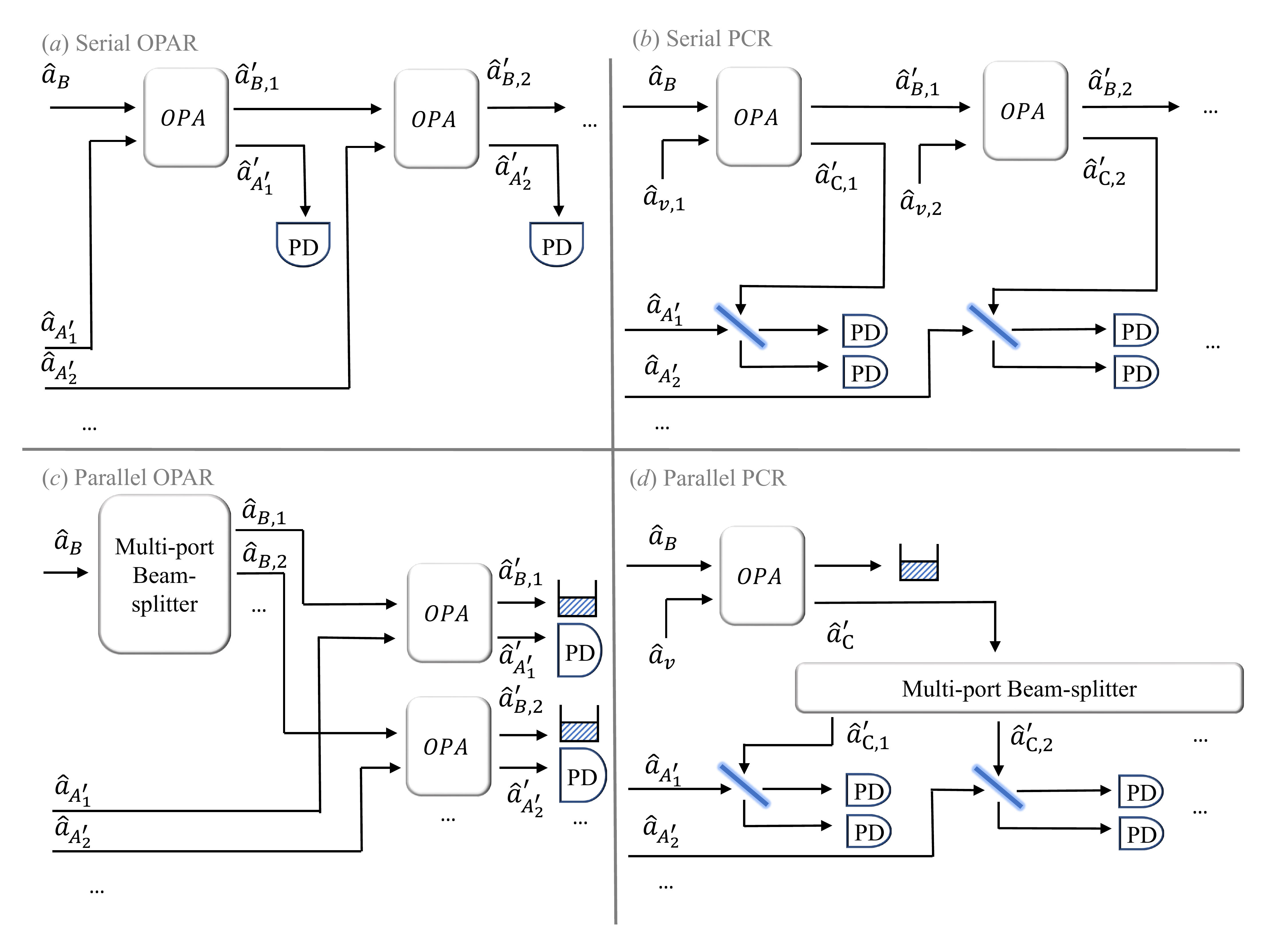}
\caption{{\bf Schematic of four receiver designs.} (a) serial optical-parametric-amplifier receiver (sOPAR) (b) serial phase-conjugate receiver (sPCR) (c) parallel optical-parametric-amplifier receiver (pOPAR) (d) parallel phase-conjugate receiver  (pPCR). 
\label{fig:receivers_main}
}
\end{figure*}

Similar to the single-sender single-receiver case, to encode a bit of information $m_k=0,1$, the $k$th sender performs a phase modulation on the signal part of the TMSV pairs via a unitary
$
\calE_{m_k}=e^{im_k \pi \hat a_{A_k}\d\hat a_{A_k}^{~}}
$
to produce the quantum system $A_k$ input to the MAC, while the idler part of the TMSV pair $A_k^\prime$ is pre-shared to the receiver side for EA. Here we have considered the binary phase-shift keying: the $k$th sender sends the bit message $m_k=0,1$ by the same probability $p_{0}=p_{1}=1/2$. To enable efficient decoding, we consider $N_{\rm R}$ repetition of such encoding---each message is repeatedly encoded on $N_{\rm R}$ signal modes of a single sender.

The decoding process takes the output of the MAC $\hat{a}_B$ and the EA idlers $\{\hat{a}_{A_k^\prime}, 1\le k \le s\}$ to decode the information $\{m_k, 1\le k \le s\}$ of all the senders.
Below we propose four receiver designs for the decoding. The basic element in the receiver design is the optical parametric amplifier (OPA), which upon input modes $\hat{a}_R$ and $\hat{a}_I$, produces two modes
$
\hat{a}_R^\prime=\sqrt{G}\hat{a}_R+\sqrt{G-1}\hat{a}_{I}^\dagger,
\hat{a}_{I}^\prime=\sqrt{G}\hat{a}_{I}+\sqrt{G-1}\hat{a}_{R}^\dagger,
$
where $G$ is the gain of the OPA. An OPA transforms the phase-sensitive correlation between the input mode-pair into the photon number difference  $\Delta\expval{\hat{a}_{I}^{\prime\dagger}\hat{a}_{I}^\prime}\propto \sqrt{G(G-1)}2{\rm Re}\expval{\hat{a}_I\hat{a}_R}$, which is widely utilized to design receivers in EA applications, such as quantum illumination~\cite{guha2009} and the bipartite EA classical communication~\cite{shi2020practical}. Moreover, one can use an OPA as a phase-conjugator to design a phase-conjugate receiver (PCR), as explained in Ref.~\cite{shi2020practical}.

To decode all $s$ messages, one can apply two different strategies, either decode them in a serial manner or in parallel. One can also base the receiver design on the direct OPA or on the phase-conjugation mechanism. These choices lead to four receiver designs---serial OPA receiver (OPAR), serial PCR, parallel OPAR and parallel PCR---as we summarize below (see details in Appendix~\ref{supp:note1}).

In the serially connected scheme, on the $k$th round, the signal output $\hat{a}_{B_{k-1}}^\prime$ from the $(k-1)$th round and the idler $\hat{a}_{A_k}$ are input to an OPA. The idler mode output from the OPA is detected, by direct detection in serial OPAR or an interferometric detection in serial PCR, to decode the message from the $k$th sender. Meanwhile, the signal mode output from the OPA is further utilized in the next round. Note that after the $k$th round, the cross correlation between the signal mode with the other idler modes are almost intact; therefore, performing an OPA on the signal and another idler $\hat{a}_{A_k^\prime}$, one can decode the message from the $k^\prime$th sender. Iterating this procedure on the remaining mode consecutively, one obtains a serial architecture for the receiver, as shown in Fig.~\ref{fig:receivers_main} (a)(b) for the serial OPAR and serial PCR.

We can also adopt a parallel design for the receivers. As the thermal-loss channel in the MAC adds excess noise into the output, we expect that in the noisy case, splitting the received signal into $s$ copies, each for the decoding of the message of a single sender, will still provide similar signal-to-noise ratios (SNR), when compared to the case without the splitting. 
In this way, each portion of the received signal can be utilized in parallel, in each individual OPA component in the parallel OPAR or in each phase-conjugation detection in the parallel PCR, to decode each message. As shown in Fig.~\ref{fig:receivers_main} (c)(d), we can design parallel-OPAR and parallel-PCR schemes.

Finally, we specify the choices of the gain in the OPA. Optimized with respect to the SNR, for OPAR the gains of the $s$ OPAs are to be $G_k=\sqrt{N_{{\rm S},k}}/\sqrt{N_{\rm B}(1+N_{\rm B})}$, $1\le k \le s$. For PCR the optimal gain turns out to be at infinity; however, we find that the performance is saturated when $(G_k-1)N_{\rm B}\gg N_{{\rm S},k}$ for the $k$th sender, thus we choose a feasible value accordingly.

\subsubsection{Receiver rate region evaluations}

\begin{figure}[tbp]
    \centering
    \includegraphics[width=0.4\textwidth]{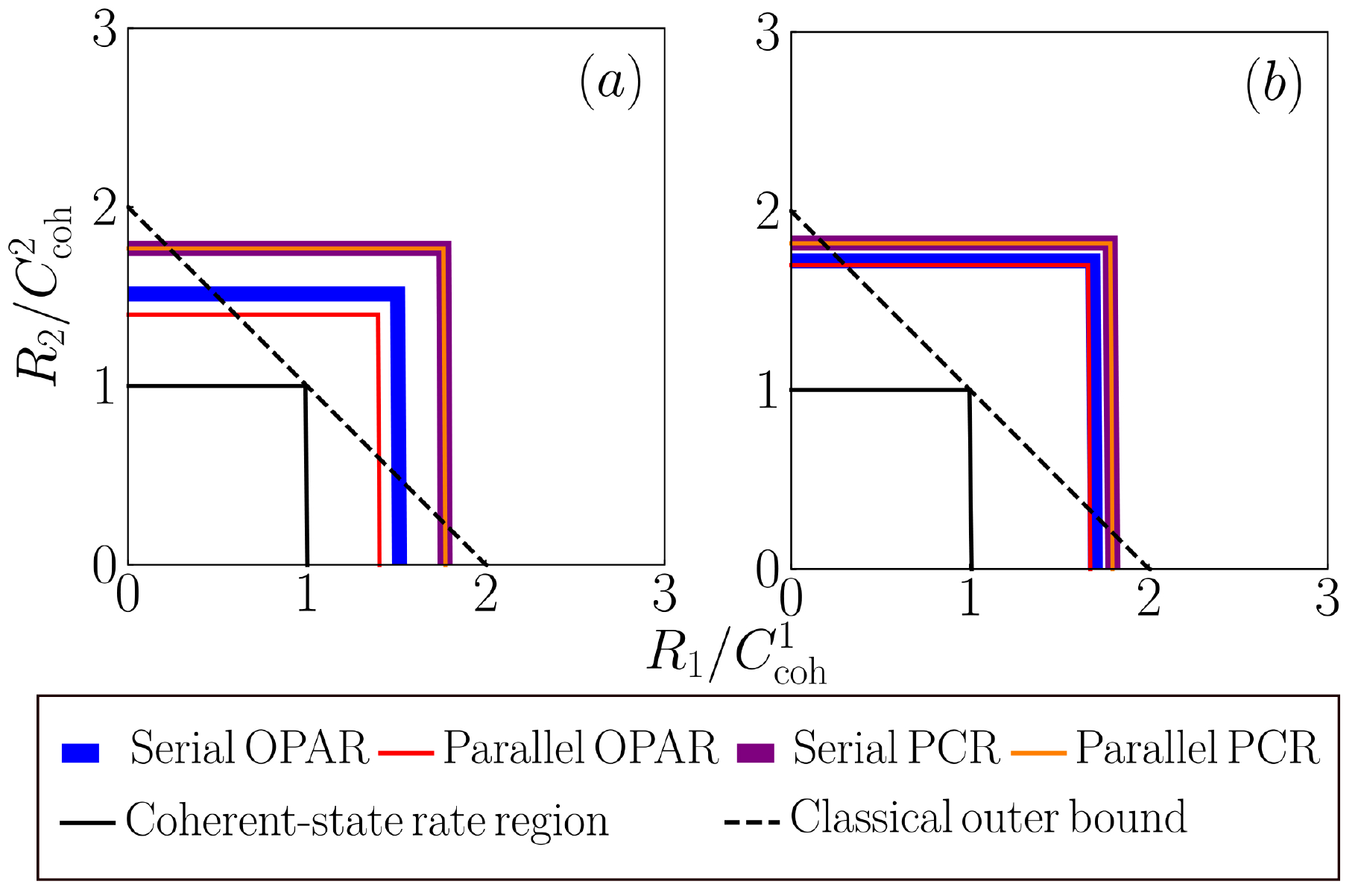}
    \caption{{\bf The two-sender rate region of our four receivers.} The rates are normalized by the coherent state bound $C_{\rm coh}^1$, $C_{\rm coh}^2$ defined in Ineq.~\eqref{C_J_region}: (a) microwave domain $N_{{\rm S},1}=N_{{\rm S},2}=0.01,\tau=0.01,N_{\rm B}=20,\eta_1=\eta_2=1/2, N_{\rm R}=2\times 10^4$;
    (b) a noisy channel with $N_{{\rm S},1}=N_{{\rm S},2}=10^{-3},\tau=10^{-3},N_{\rm B}=10^4,\eta=1/2, N_{\rm R}=10^7$. To distinguish between the overlapping lines, we plot the serial receivers in thicker lines by contrast with the parallel receivers plotted narrowed.
    The gains of OPAR are given in the main text, and the gains of PCR are $G=2$ for $N_{\rm B}=20$ and $G=1+10^{-3}$ for $N_{\rm B}=10^4$. We also compare the receiver rate region with the coherent-state rate (black solid) region and the classical outer bound (black dashed).
    \label{fig:rateregion}
    }
\end{figure}

\begin{figure}[tbp]
    \centering
    \includegraphics[width=0.475\textwidth]{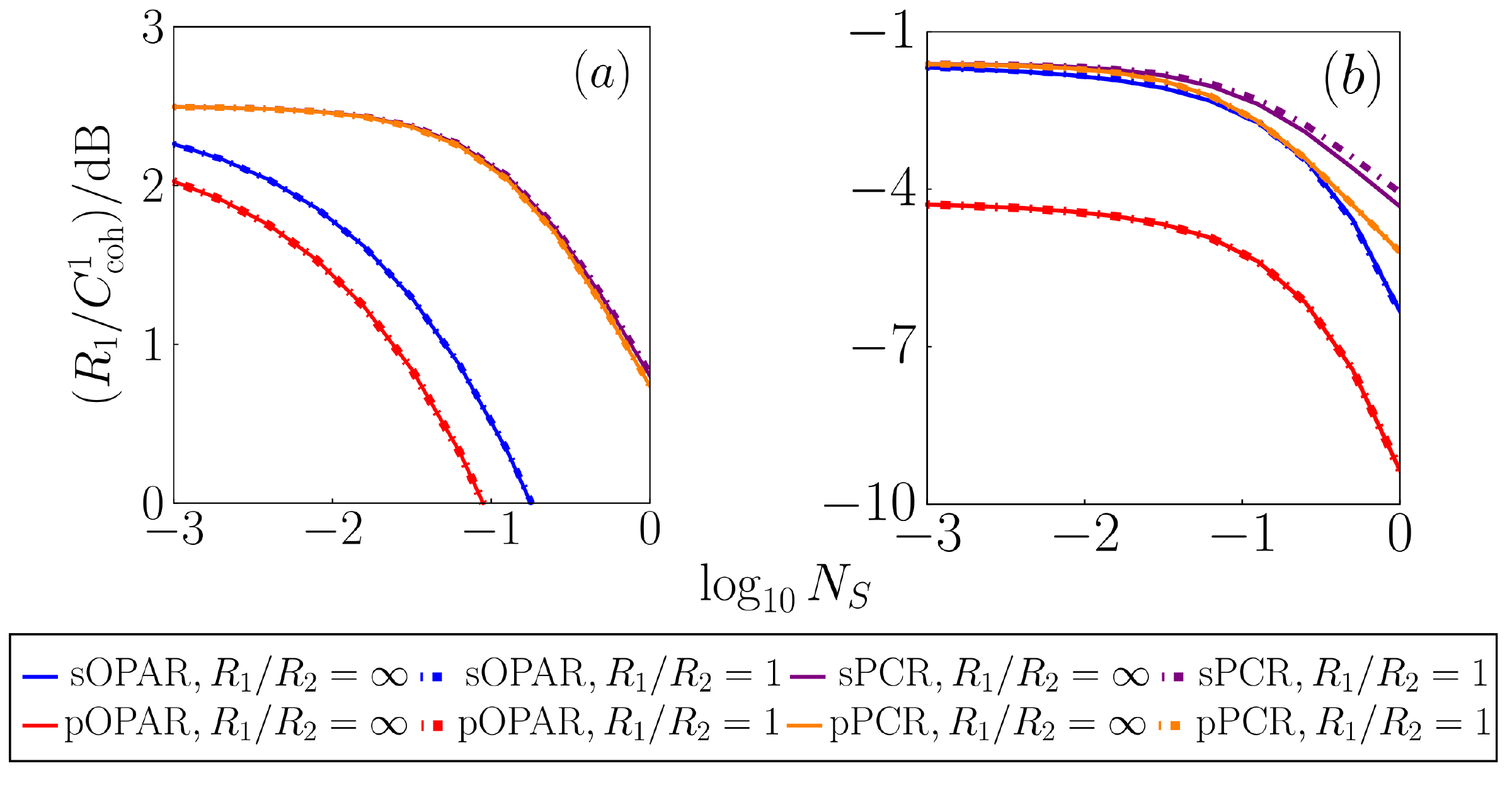}
    \caption{{\bf Receiver rates versus signal brightness.} The dependence on source brightness $N_{{\rm S},1}=N_{{\rm S},2}=N_{\rm S}$ of the EA advantage of the four receivers for two-sender MAC communication under the scenario of (a) microwave domain $\eta=1/2,\tau=0.01,N_{\rm B}=20$; (b) long wavelength infrared domain $\eta_1=\eta_2=1/2,\tau=0.001,N_{\rm B}=0.1$. In the legend s, p refer to `serial' and `parallel' respectively. The number of modes $N_{\rm R}$ is fixed such that the SNR $N_{\rm R}\tau N_{\rm S}/N_{\rm B}=0.1$ for sender $i=1,2$.  We plot $R_1$ for sender 1 under conditions $R_1/R_2=\infty$ (solid) and $R_1/R_2=1$ (dot-dashed). Note that $R_1/R_2=0,\infty$ are equivalent up to a swap due to the symmetry between the two senders; and for given $R_1/R_2$, $R_2/C^2_{\rm coh}=R_1/C^1_{\rm coh}$. 
    }
    \label{fig:Ns}
\end{figure}

As the encoding and receivers are chosen, now the (soft-decoding) rate region is entirely obtained from the classical formula of conditional mutual information~\cite{cover1999elements} computed over the measurement outcome distribution (see Appendix~\ref{supp:note1}). 
As shown in Fig.~\ref{fig:rateregion}, we compare the receiver rate regions with the classical coherent-state rate region in Ineq.~\eqref{C_J_region} (black solid) and the classical outer bound in Ineq.~\eqref{C_outer_bound} jointly and Ineq.~\eqref{C_J_region} (black dashed) with $J$ being all senders. We see that the performance of both OPAR and PCR can beat the classical coherent state rate region and the classical outer bound.

In Fig.~\ref{fig:rateregion}, the performance of the OPAR (blue solid and red solid) is inferior to the PCR (purple solid and orange solid), a gap between which is significant in Fig.~\ref{fig:rateregion}(a). This is because the PCR has a better SNR to the next order in $N_{\rm S}$ compared with OPA, as found in the single-sender case in Ref.~\cite{shi2020practical} and confirmed in Fig.~\ref{fig:Ns} here. As the brightness $N_{\rm S}$ decreases in Fig.~\ref{fig:rateregion}(b), the gap between the PCR and OPAR almost diminishes. In Fig.~\ref{fig:Ns}, we see the rates of OPAR (blue and red) are lower than PCR (purple and orange), with a gap expanding as the brightness $N_{\rm S}$ grows. We also find that the rate advantage of both the OPAR and PCR saturates to 3dB as the brightness $N_{\rm S}$ decreases, consistent with the SNR advantage in quantum illumination~\cite{Tan2008,zhuang2017}. This is because when the noise is large, the information rate is proportional to the SNR. Note that when the channel noise $N_{\rm B}$ decreases, the theoretical EA advantage evaluated by TMSV remains substantial. However, the practical advantage allowed by our receivers diminishes as $N_{\rm B}$ falls below 1.  For $N_{\rm B}=0.1$ (and smaller), there is no advantage for the proposed receivers, as shown in Fig.~\ref{fig:Ns}(b). This leaves an open question that a feasible receiver that provides EA advantage in the low-noise scenario is hitherto elusive.

\section{Discussion}
In this paper, we have solved the capacity region of entanglement-assisted classical communication over a quantum multiple-access channel with an arbitrary number of senders. We also provide explicit encoding and decoding strategies that offer a practical route towards achieving quantum advantages in such network communication scenarios. Due to teleportation~\cite{bennett1993teleporting} and super-dense coding~\cite{bennett1992}, the rate region of EA quantum communication is precisely half of the EA classical communication region; therefore, all of our results can be straightforwardly extended to the case of quantum communication. The explicit protocols can also be used for EA quantum communication via further combining with a teleportation protocol.

Many future directions can be explored. For example, multi-partite entanglement may be considered instead of the product form of Eq.~\eqref{EA_state} to assist the communication scenario, when the senders can collaborate in the entanglement distribution process. Another open question is whether one can have superadditivity  phenomena in our entanglement-assisted capacity region of multiple-access channels.

Before closing, we discuss potential experimental realizations for the proposed EA-MAC communication systems. The basic setup will be similar to that in Ref.~\cite{hao2021entanglement}, with entanglement generated by spontaneous parametric down-conversion in a nonlinear crystal. The receiver can be implemented with another nonlinear crystal to perform phase conjugation or parametric amplification. However, the challenge to demonstrate an entanglement advantage under the multiple-senders scenario is that the pump beams for different entanglement sources need to be frequency and phase locked. Moreover, each stored idler needs to be phase locked to its corresponding signal received from the MAC. Differential-phase encoding can potentially avoid the need for phase locking, which is subject to future studies.

\begin{appendix}

\section{Details of formalism}
\label{app:method}
\subsection{Formal analysis of the EA MAC}
In the MAC communication scenario of Fig.~\ref{fig:EAMACschematic_main}, each encoded signal-idler is in a state $\hat{\sigma}_{A_kA_k^\prime}^{m_k}=\calE_{m_k}\otimes \calI \left(\hat{\phi}_{A_kA_k^\prime}\right)$, where $\calI$ is the identity channel modeling the ideal storage of the idler system. Denote the overall encoding operation as $\calE_m=\otimes_{k=1}^s \calE_{m_k}$, and the probability of sending each message as $p_m=\prod_{k=1}^s p_{m_k}$, the overall encoding can be described by the composite quantum state
\be 
\hat{\Xi}_{MAA^\prime}=\sum_m p_m \state{m}_M \otimes \hat{\sigma}^m_{AA^\prime},
\label{eq:Xi}
\ee
where
$
\hat{\sigma}^m_{AA^\prime}=\otimes_{k=1}^s \hat{\sigma}_{A_kA_k^\prime}^{m_k}\equiv
\left[\calE_{m} \otimes \calI\right] \left(\hat{\phi}_{AA^\prime}\right)
$
is the overall encoded state conditioned on message $m$ and $\state{m}_M$ is the classical register for message $m$.

After the encoding, all of the quantum systems from the $s$ senders $A$ are input to the MAC $\calN_{A\to B}$, which outputs the quantum system $B$ for the receiver to decode the messages jointly with the EA $A^\prime$. The overall state after the channel is
\be 
\hat{\omega}_{MBA^\prime}=\sum_m p_m \state{m}_M \otimes \hat{\beta}^m_{BA^\prime},
\label{eq:omega}
\ee 
where $\hat{\beta}^m_{BA^\prime}=\left[\calN_{A\to B}  \otimes \calI\right] \left(\hat{\sigma}^m_{AA^\prime}\right)$.

The performance metric of a communication scenario over a MAC is described by a vector of rates $(R_1,\cdots,R_s)$, where $R_k$ is the reliable communication rate between the $k$th sender and the receiver. These rates in general have non-trivial trade-offs with each other.
Formally, we define an $(n,R_1,\cdots,R_s,\epsilon)$ EA code by: the prior set $\{p_{\bm m_k^n}\}$, the encoded quantum states $\{\hat\sigma^{\bm m_k^n}\}$, $1\le k \le s$ on $\bm A^n$, with each message $\bm m_k^n\in [2^{nR_k}]$, and the decoding positive operator-valued measure (POVM) $\{\hat{\Lambda}_{m_1\cdots m_s}\}$ on $\bm B^n \bm A^{\prime n}$ such that 
\be 
{\rm Tr}\Bigg\{\hat{\Lambda}_{m_1\cdots m_s}\left[\calN^{\otimes n}\circ \left(\otimes_{k=1}^s\calE_{m_k}\right)\otimes \calI_{\bm A^{\prime n }} \right] \hat{\phi}_{\bm A^n \bm A^{\prime n}} \Bigg\}\ge 1-\epsilon.
\ee 
We say that $(R_1,\cdots,R_s)$ is an achievable rate vector if for all $\epsilon>0,\delta>0$ and sufficiently large $n$, there exists an $(n,R_1-\delta,\cdots,R_s-\delta,\epsilon)$ EA code. The EA classical capacity region $\calC_{\rm E}(\calN)$ is defined to be the closure of the set of all achievable rate vectors. The regularized capacity $\calC_{\rm E}(\calN)$ is the union of all $\ell$-letter one-shot capacity regions $\calC_{\rm E}^{(1)}(\calN^{\otimes \ell})/\ell$, with integers $\ell\ge 1$. The one-shot capacity region $\calC_{\rm E}^{(1)}(\calN)$ is the closure of the subset of achievable rate vectors by $(n,R_1-\delta,\cdots,R_s-\delta,\epsilon)$ codes that are generated from separable inputs among the $n$ channel uses. Here ``one-shot'' is in the sense that the entanglement is constrained in a single channel use. In this regard, the capacity region $\calC_{\rm E}^{(1)}(\calN^{\otimes \ell})$ considers $\calN^{\otimes \ell}$ as a single channel and allows codes with entanglement between $\ell$ uses of $\calN$.
In the case without EA, the capacity region is well-established by the pioneering work of Winter~\cite{winter2001capacity} (see Appendix~\ref{supp:note2}).

\subsection{ Outer bounds for bosonic thermal-loss MAC}
Now we provide the outer bound in Ineqs.~\eqref{eq:EA_outer_bound} for the EA classical capacity region of the bosonic thermal-loss MAC. As we see in Fig.~\ref{fig:MACbosonic}, the overall channel can be written as a concatenation of two parts, $\calN=\calL^{\tau,N_{\rm B}}\circ \calE_{\rm MAC}$, where $\calE_{\rm MAC}$ represents the beamsplitter modeling the signal interference. 
From the bottleneck inequality, the overall communication rate is upper bounded by 
\be 
\sum_{k=1}^s R_k \le C_{\rm E}\left(\sum_{k=1}^s \eta_k N_{{\rm S},k}; \calL^{\tau, N_{\rm B}}\right),
\ee 
the single-sender single-receiver EA classical capacity of the thermal-loss channel $\calL^{\tau,N_{\rm B}}$ with brightness $\sum_{k=1}^s \eta_k N_{{\rm S},k}$. This is because for the channel $\calL^{\tau,N_{\rm B}}$, only a single mode signal $\hat{a}_{A_{\rm mix}}$ in Eq.~\eqref{Amix} with brightness $\expval{\hat{a}_{A_{\rm mix}}^\dagger \hat{a}_{A_{\rm mix}}}=\sum_{k=1}^s \eta_k N_{{\rm S},k}$ goes through. 
Explicitly, the capacity  
\be
C_{\rm E}\left(N_{\rm S}; \calL^{\tau, N_{\rm B}}\right)=g(N_{\rm S})+g(N_{\rm S}^\prime)-g(A_+)-g(A_-),
\label{CE_full}
\ee
with 
$A_\pm=(D-1\pm(N_{\rm S}^\prime-N_{\rm S}))/2$, $N_{\rm S}^\prime=\tau N_{\rm S}+N_{\rm B}$ and $D=\sqrt{(N_{\rm S}+N_{\rm S}^\prime+1)^2-4\tau N_{\rm S}(N_{\rm S}+1)}$. This proves Ineq.~\eqref{eq:EA_outer_bound_sum}.

As for the individual upper bounds for the senders in Ineq.~\eqref{eq:EA_outer_bound_individual}, we consider a theoretical super-receiver with access to all of the output ports of the beamsplitter part. The super-receiver performs the reverse of the beamsplitter transform, after which the communication reduces to the single-sender scenario of which the information rate is bounded by each single-sender single-receiver EA classical capacity. 
Explicitly, we have
\bal
R_k&\leq C_{\rm E}\left( N_{{\rm S},k}; \calL^{\tau, N_{\rm B}}\right), 1\le k \le s,
\eal
which proves Ineq.~\eqref{eq:EA_outer_bound_individual}.

\section{Analyses of the receiver designs}
\label{supp:note1}

To begin with, we briefly summarize the receiver problem to be solved in this section. Note that in this section, we will write subscripts inside subscripts as normal texts so that the equations are not too small to be visible. In the phase encoding scheme, each sender applies a phase rotation on the signal mode $\hat a_{A_k}$ of its TMSV source, which modules the signal-idler correlations by
$\expval{\hat{a}_{A_k}\hat{a}_{A_k^\prime}}=C_{pk}\equiv\sqrt{ N_{{\rm S},k}(N_{{\rm S},k}+1)}e^{i\theta_k}$. For binary encoding, each $\theta_k$ has two possible values $\theta_k^{(0)}=0,\theta_k^{(1)}=\pi$. The signal modes $\hat a_{A_k}$ then goes through the thermal-loss MAC, which produces the received mode mode $\hat{a}_B$ given by Eq.~\eqref{Amix} and Eq.~\eqref{thermal_loss} of the main paper.

Below we assess the receivers in the two-sender scenario as an example. The $s$-sender case is solved in the same way. For convenience, let $\eta_1=\eta,\eta_2=1-\eta$. To ease the reading, we mark the modes after the OPAs with the superscript `$\prime$'.

In the BPSK encoding of all senders, denote the phases of the $s$ senders as $\bm \theta^{(m)}$ conditioned on the message $m$. The phase-modulated photon statistics $P(\bm n|\bm \theta^{(m)})$ of the measurement is derived in each subsections. It is convenient to define the photon statistics with respect to the message $m$: $P(\bm n|m)\equiv P(\bm n|\bm \theta^{(m)})$.
Define the measurement register space of system $B$ as $B_R$, then the conditional Shannon information $I(M[J];B_R|M[J^c])$ can be obtained from the measurement statistics $P(\bm n|m)$
\begin{widetext}
\begin{align}
I(M[J];&B_R|M[J^c])=
\sum_{m[J^c]}p_{m[J^c]}\Bigg\{H\left[\left\{\sum_{m\left[J\right]}p_{m\left[J\right]}P\left(\bm n|m\right)\right\}_{\bm n}\right]
-\sum_{m\left[J\right]}p_{m\left[J\right]} H\left[\left\{P\left(\bm n|m\right) \right\}_{\bm n}\right]\Bigg\},
\label{I_receiver}
\end{align}
\end{widetext}
where $H[\{P(\bm n)\}_{\bm n}]=-\sum_{\bm n} P(\bm n)\log P(\bm n)$ is the Shannon entropy of the distribution $P$. Different from the quantity $I(M[J];B|M[J^c])$, here Eq.~\eqref{I_receiver} is the formula of the optimum information rate for a specific measurement strategy given by one of the four receiver designs.

While the performances of OPAR are evaluated exactly, the performances of PCR are evaluated with a Gaussian approximation on the photon statistics, which is precise when the number of repetition modes $N_{\rm R}$ is large.

\subsection{Serial OPAR (Fig.~\ref{fig:receivers_main}a of the main paper)}

For the first OPA with a gain $G_1$,
\bal
\hat a_{B,1}^{\prime}&=\sqrt{G_1}\hat a_B+\sqrt{G_1-1}\hat a_{A^\prime_1}^\dagger,\\
\hat a_{A^\prime_1}^\prime&=\sqrt{G_1}\hat a_{A^\prime_1}+\sqrt{G_1-1}\hat a_{B}^\dagger
\,.\eal
We measure the photon counts of $N_{\rm R}$ independent and identical (i.i.d.) copies of $\hat a_{A^\prime_1}\p$.

For the second OPA with a gain $G_2$,
\bal
\hat a_{B,2}^{\prime}&=\sqrt{G_2}\hat a_{B,1}^{\prime}+\sqrt{G_2-1}\hat a_{A^\prime_2}^\dagger\\
\hat a_{A^\prime_2}^\prime&=\sqrt{G_2}\hat a_{A^\prime_2}+\sqrt{G_2-1}\hat a_{B,1}^{\prime\dagger}
\,.\eal
We measure the photon counts of $N_{\rm R}$ i.i.d. copies of $\hat a_{A^\prime_2}\p$.

We make decision on the $N_{\rm R}$ i.i.d photon counts of modes $\hat a_{A^\prime_1}^\prime, \hat a_{A^\prime_2}^\prime$. Indeed, $\hat a_{A^\prime_1}^\prime, \hat a_{A^\prime_2}^\prime$ are in a zero-mean Gaussian state with the covariance matrix
\begin{widetext}
\begin{align}
V\equiv
\expval{ \bp
\hat a_{A^\prime_1}^\prime\\
\hat a_{A^\prime_1}^{\prime\dagger}\\
\hat a_{A^\prime_2}^\prime\\
\hat a_{A^\prime_2}^{\prime\dagger}\\
\ep
\bp
\hat a_{A^\prime_1}^{\prime\dagger}& \hat a_{A^\prime_1}^\prime& \hat a_{A^\prime_2}^{\prime\dagger}& \hat a_{A^\prime_2}^\prime
\ep
}
=\bp
a+1 & 0 & c & 0\\
0 & a & 0 & c^\star\\
c^\star & 0 & s+1 & 0\\
0 & c & 0 & s
\ep,
\label{eq:OPACM}
\end{align}
where the constants
\bal 
a=&G_1 N_{{\rm S},1} + 2 \sqrt{\left(G_1-1\right)G_1\tau\eta} \Re C_{p1}  + \left(G_1-1\right) \left[1+N_B + N_S^\star\right]\,,\\
s=&G_2 N_{{\rm S},2} + 2 \sqrt{\left(G_2-1\right)G_2G_1\tau\left(1-\eta\right)} \Re C_{p2}  \\
&+ \left(G_2-1\right) \Big\{\left(G_1-1\right)N_{{\rm S},1} +G_1\big[1+N_B + N_S^\star\big]+2\Re \sqrt{\left(G_1-1\right)G_1\tau\eta}C_{p1}\Big\}\,,\\ 
c=& \sqrt{\left(G_2-1\right)}\Big[\sqrt{\left(G_1-1\right)G_1} \left(N_S^\star+1 + N_{{\rm S},1}\right)+\left(G_1-1\right)\sqrt{\tau\eta}C_{p1}^\star+G_1\sqrt{\tau\eta}C_{p1}\Big]+C_{p2}^\star \sqrt{\tau\eta G_2\left(G_1-1\right)}\,.\\
\eal
\end{widetext}
Here $N_S^\star= \tau\left[ N_{{\rm S},1} \eta +N_{{\rm S},2} \left(1-\eta\right)\right]+N_B$. The dependence on the phases $\theta_1,\theta_2$ lies in the phase-sensitive correlations $C_{p1},C_{p2}$.

Given $V$, we immediately have the covariance matrix of the quadratures $\hat q_{A^\prime_k}^\prime=\hat a_{A^\prime_k}^\prime+\hat a_{A^\prime_k}^{\prime\dagger}$, $\hat p_{A^\prime_k}^\prime=-i\left(\hat a_{A^\prime_k}^\prime-\hat a_{A^\prime_k}^{\prime\dagger}\right)$ for the two modes $k=1,2$ 
\be
V_{\rm quad}\equiv
\expval{
\bp
\hat q_{1}\\
\hat p_{1}\\
\hat q_{2}\\
\hat p_{2}\\
\ep
\bp
\hat q_{1}& \hat p_{1}& \hat q_{2}& \hat p_{2}
\ep
}
=\bp
E & 0 & C & 0\\
0 & E & 0 & C\\
C & 0 & S & 0\\
0 & C & 0 & S
\ep
,
\label{eq:CMqp}
\ee
which can be obtained via the relationship $V_{\rm quad}=UVU^\dagger$ with the transform matrix $U$
\be
U=
\bp
1 & 1 & 0 & 0\\
-i & i & 0 & 0\\
0 & 0 & 1 & 1\\
0 & 0 & -i & i
\ep
\,.\ee

The probability distribution of the random photon counts of such a two-mode Gaussian state with the quadrature covariance matrix Eq.~\eqref{eq:CMqp} is given by
\bal
P(n_1,n_2|\theta_1,\theta_2)=& -4 F_R(1+n_1,1+n_2,1,\frac{4C^2}{XY})
\\
&\times \frac{(-1+C^2+E+S-E S)^{1+n_1 +n_2}}{X^{1+n_1}Y^{1+n_2}},
\label{eq:2dPD}
\eal
where $F_R$ is the regularized hypergeometric function and 
$
X=1+C^2+E-(1+E)S,
Y=C^2-(E-1)(S+1).
$
The information rates are then obtained via Eq.~\eqref{I_receiver}.

\subsection{Serial PCR (Fig.~\ref{fig:receivers_main}b of the main paper)}
For the first OPA with a gain $G_1$,
\bal
\hat a_{B,1}^{\prime}&=\sqrt{G_1}\hat a_B+\sqrt{G_1-1}\hat a_{v,1}\d,\\
\hat a_{C,1}^{\prime}&=\sqrt{G_1}\hat a_{v,1}+\sqrt{G_1-1}\hat a_B\d
\,,
\eal
where the ancilla $\hat{a}_{v,1}$ is in a vacuum mode.
The following balanced beamsplitter yields two arms
$
\hat a_{X,1}=(\hat a_{C,1}^{\prime}+\hat a_{A_1^\prime})/\sqrt{2},
\hat a_{Y,1}=(\hat a_{C,1}^{\prime}-\hat a_{A_1^\prime})/\sqrt{2}.
$
We measure the photon count differences of $N_{\rm R}$ i.i.d. copies between $\hat a_{X,1}$ and $\hat a_{Y,1}$.

For the second OPA with a gain $G_2$,
\bal
\hat a_{B,2}^{\prime}&=\sqrt{G_2}\hat a_{B,1}^{\prime}+\sqrt{G_2-1}\hat a_{v,2}\d,\\
\hat a_{C,2}^{\prime}&=\sqrt{G_2}\hat a_{v,2}+\sqrt{G_2-1}\hat a_{B,1}^{\prime\dagger}
\,.\eal
Similarly the following balanced beamsplitter yields two arms
$
\hat a_{X,2}=(\hat a_{C,2}^{\prime}+\hat a_{A_2^\prime})/\sqrt{2},
\hat a_{Y,2}=(\hat a_{C,2}^{\prime}-\hat a_{A_2^\prime})/\sqrt{2}.
$
We measure the photon count differences between $\hat a_{X,2}$ and $\hat a_{Y,2}$ of $N_{\rm R}$ i.i.d. copies. The modes $\hat a_{X,1}, \hat a_{Y,1}, \hat a_{X,2}, \hat a_{Y,2}^\prime$ are in a zero-mean Gaussian state with the covariance matrix 
\begin{widetext}
\bal
V&\equiv
\expval{
\bp
\hat a_{X,1}^\prime\\
\hat a_{X,1}^{\prime\dagger}\\
\hat a_{Y,1}^\prime\\
\hat a_{Y,1}^{\prime\dagger}\\
\hat a_{X,2}^\prime\\
\hat a_{X,2}^{\prime\dagger}\\
\hat a_{Y,2}^\prime\\
\hat a_{Y,2}^{\prime\dagger}\\
\ep
\bp
\hat a_{X,1}^{\prime\dagger}& \hat a_{X,1}^\prime& \hat a_{Y,1}^{\prime\dagger}& \hat a_{Y,1}^\prime& \hat a_{X,2}^{\prime\dagger}& \hat a_{X,2}^\prime& \hat a_{Y,2}^{\prime\dagger}& \hat a_{Y,2}^\prime
\ep
}
=\left(
\begin{array}{cccccccc}
 {a_1}+1 & 0 & {c_1} & 0 & {a_3} & 0 & {c_{31}} & 0 \\
 0 & {a_1} & 0 & {c_1}^\star & 0 & {a_3}^\star & 0 & {c_{31}}^\star \\
 {c_1}^\star & 0 & {s_1}+1 & 0 & {c_{32}} & 0 & {s_3} & 0 \\
 0 & {c_1} & 0 & {s_1} & 0 & {c_{32}}^\star & 0 & {s_3}^\star \\
 {a_3}^\star & 0 & {c_{32}}^\star & 0 & {a_2}+1 & 0 & {c_2} & 0 \\
 0 & {a_3} & 0 & {c_{32}} & 0 & {a_2} & 0 & {c_2}^\star \\
 {c_{31}}^\star & 0 & {s_3}^\star & 0 & {c_2}^\star & 0 & {s_2}+1 & 0 \\
 0 & {c_{31}} & 0 & {s_3} & 0 & {c_2} & 0 & {s_2} \\
\end{array}
\right),
\label{eq:PCRCM}
\eal
where the constants
\bal
a_1&=\frac{1}{2} \left(2 \Re({C_{p1}}) \sqrt{\eta  ({G_1}-1) \tau }+\eta  {G_1} \tau  {N_{{\rm S},1}}-(\eta -1) ({G_1}-1) \tau  {N_{{\rm S},2}}-\eta  \tau  {N_{{\rm S},1}}+({G_1}-1) ({N_{B}}+1)+{N_{{\rm S},1}}\right)\,,\\
s_1&=\frac{1}{2} \left(-2 \Re({C_{p1}}) \sqrt{\eta  ({G_1}-1) \tau }+\eta  ({G_1}-1) \tau  {N_{{\rm S},1}}-(\eta -1) ({G_1}-1) \tau  {N_{{\rm S},2}}+({G_1}-1) ({N_{B}}+1)+{N_{{\rm S},1}}\right)\,,\\
c_1&=\frac{1}{2} \left(-2 i \Im({C_{p1}}) \sqrt{\eta  ({G_1}-1) \tau }+\eta  (1-{G_1}) \tau  {N_{{\rm S},1}}+(\eta -1) ({G_1}-1) \tau  {N_{{\rm S},2}}+(1-{G_1}) ({N_{B}}+1)+{N_{{\rm S},1}}\right)\,,\\
a_3&=\frac{1}{2} \left({C_{p2}}^\star \sqrt{(1-\eta ) ({G_1}-1) \tau }+{C_{p1}} \sqrt{\eta  {G_1} ({G_2}-1) \tau }+\sqrt{({G_1}-1) {G_1} ({G_2}-1)} (N_S^\star+1)\right)\,,\\
s_3&=\frac{1}{2} \left(-{C_{p2}}^\star \sqrt{(1-\eta ) ({G_1}-1) \tau }+{C_{p1}} \left(-\sqrt{\eta  {G_1} ({G_2}-1) \tau }\right)+\sqrt{({G_1}-1) {G_1} ({G_2}-1)} (N_S^\star+1)\right)\,,\\
c_{31}&=\frac{1}{2} \left({C_{p2}}^\star \sqrt{(1-\eta ) ({G_1}-1) \tau }+{C_{p1}} \left(-\sqrt{\eta  {G_1} ({G_2}-1) \tau }\right)-\sqrt{({G_1}-1) {G_1} ({G_2}-1)} (N_S^\star+1)\right)\,,\\
c_{32}&=\frac{1}{2} \left(-{C_{p2}}^\star \sqrt{(1-\eta ) ({G_1}-1) \tau }+{C_{p1}} \sqrt{\eta  {G_1} ({G_2}-1) \tau }-\sqrt{({G_1}-1) {G_1} ({G_2}-1)} (N_S^\star+1)\right)\,,\\
a_2&=\frac{1}{2} \left(2 \Re({C_{p2}}) \sqrt{(1-\eta ) {G_1} ({G_2}-1) \tau }+{G_1} ({G_2}-1) (N_S^\star+1)+{N_{{\rm S},2}}\right)\,,\\
s_2&=\frac{1}{2} \left(-2 \Re({C_{p2}}) \sqrt{(1-\eta ) {G_1} ({G_2}-1) \tau }+{G_1} ({G_2}-1) (N_S^\star+1)+{N_{{\rm S},2}}\right)\,,\\
c_2&=\frac{1}{2} \left(-2 i \Im({C_{p2}}) \sqrt{(1-\eta ) {G_1} ({G_2}-1) \tau }-{G_1} ({G_2}-1) (N_S^\star+1)+{N_{{\rm S},2}}\right)\,.\\
\eal
\end{widetext}
The dependence on phase modulation $\theta_1,\theta_2$ lies in the phase-sensitive correlations $C_{p1},C_{p2}$. 

As it is challenging to numerically evaluate the photon count distribution of a four-mode Gaussian state, we estimate the performance of PCRs by a Gaussian approximation. From the covariance matrix $V$ of $\hat a_{X,1}, \hat a_{Y,1}, \hat a_{X,2}, \hat a_{Y,2}^\prime$, by Wick's theorem we immediately have the $2\times 2$ covariance matrix $V_d$ of the photon differences $\Delta\hat N_{i}=\hat N_{X,i}-\hat N_{Y,i}=\expval{\hat a_{X,i}\hat a_{X,i}}-\expval{\hat a_{Y,i}\hat a_{Y,i}}$ for the two slices $i=1,2$. Due to the central limit theorem, the mean distribution of $N_{\rm R}$ i.i.d. copies of random variables $\Delta N_{1},\Delta N_{2}$ converges to the Gaussian distribution when $N_{\rm R}\to\infty$
\be
P(\bm n|\theta_1,\theta_2)=\frac{1}{(2\pi)^2 \sqrt{\text{det}\,V_d^{(N_{\rm R})}}}\exp{-\frac{\bm n^T V_d^{(N_{\rm R})-1}\bm n}{2}},
\label{eq:GausApprox}
\ee
where $V_d^{(N_{\rm R})}=V_d/\QZ{N_{\rm R}}$ is the covariance matrix of the $N_{\rm R}$-copy mean distribution, $\bm n$ is the two-dimensional output averaged over $N_{\rm R}$ copies of the random photon count differences $\Delta N_{1},\Delta N_{2}$. The information rates are then obtained via Eq.~\eqref{I_receiver}.

\subsection{Parallel OPAR (Fig.~\ref{fig:receivers_main}c of the main paper)}
At the receiver side we apply a beamsplitter to slice the \QZ{received} mode $\hat a_B$ into two copies $i=1,2$
\be
\hat a_{B,i}=\sqrt{\eta_i}\hat a_B+\sqrt{1-\eta_i}\hat a_{v,i}
\,,
\ee
where each ancilla $\hat{a}_{v,i}$ is in vacuum.
For each OPA $i=1,2$
\bal
\hat a_{B,i}^\prime&=\sqrt{G_i}\hat a_{B,i}+\sqrt{G_i-1}\hat a_{A_i^\prime}^\dagger,\\
\hat a_{A_i^\prime}^\prime&=\sqrt{G_i}\hat a_{A_i^\prime}+\sqrt{G_i-1}\hat a_{B,i}^\dagger
\,.\eal
We measure the photon counts of $N_{\rm R}$ i.i.d. copies of $\hat a_{A_1^\prime}\p$, $\hat a_{A_2^\prime}\p$. These modes are in a Gaussian state with the covariance matrix in the same form of Eq.~\eqref{eq:OPACM} valued by
\begin{widetext}
\bal
a=&\eta  (G_1-1) N_B+2 \eta  \sqrt{\left(G_1-1\right) G_1 \tau } \Re\left(C_{{p1}}\right)+\eta  \left(G_1-1\right) \tau  \left(\eta  N_{{\rm S},1}+(1-\eta ) N_{{\rm S},2}\right)+G_1 N_{{\rm S},1}+G_1-1\,,\\
s=&(G_2-1) \left[(1-\eta ) \left(N_B+\eta  \tau  N_{{\rm S},1}\right)+(1-\eta )^2 \tau  N_{{\rm S},2}\right]+G_2(N_{{\rm S},2}+1)+2 (\eta -1) \sqrt{\left(G_2-1\right) G_2 \tau } \Re\left(C_{{p2}}\right)-1\,,\\
c=&-\sqrt{(1-\eta ) \eta  \left(G_2-1\right)} \left(\sqrt{G_1-1} N_S^\star+C_{{p1}} \sqrt{G_1 \tau }\right)+C_{{p2}}^\star \sqrt{-(\eta -1) \eta  \left(G_1-1\right) G_2 \tau }\,.\\
\eal
\end{widetext}
where $N_S^\star=\tau[\eta N_{{\rm S},1}+(1-\eta )  N_{{\rm S},2}]+N_B$. The photon count distribution is given by Eq.~\eqref{eq:2dPD} with quadrature covariance matrix evaluated by plugging $a,s,c$ in Eq.~\eqref{eq:CMqp}. The information rates are then obtained via Eq.~\eqref{I_receiver}.

\subsection{Parallel PCR (Fig.~\ref{fig:receivers_main}d of the main paper)}
By contrast with the parallel OPAR, at the receiver side we first apply the phase conjugation by an OPA to obtain
\bal
\hat a_{C}^{\prime}&=\sqrt{G}\hat a_{v}+\sqrt{G-1}\hat a_B\d
\,,\eal
where the ancilla $\hat a_{v}$ is in vacuum. Then a beamsplitter slices the mode $\hat a_{C}^{\prime}$ into two copies $i=1,2$,
\bal
\hat a_{C,i}^{\prime}=\sqrt{\eta_i}\hat a_C^{\prime}+\sqrt{1-\eta_i}\hat a_{v,i}
\,.\eal
For each slice $i=1,2$, a balanced beamsplitter yields two arms
$
\hat a_{X,i}=(\hat a_{C,i}+\hat a_{A_i^\prime})/\sqrt{2},
\hat a_{Y,i}=(\hat a_{C,i}-\hat a_{A_i^\prime})/\sqrt{2}.
$
We measure the photon count differences between $\hat a_{X,i}$ and $\hat a_{Y,i}$ of $N_{\rm R}$ i.i.d. copies. Here $\hat a_{X,1}, \hat a_{Y,1}, \hat a_{X,2}, \hat a_{Y,2}^\prime$ are in a zero-mean Gaussian state with the covariance matrix in the same form of \eqref{eq:PCRCM}. The evaluation of parameters $a_1,s_1,c_1, \ldots$ is straightforward and lengthy, so we omit it here. A Gaussian approximation for the probability distribution of the measurement results follows from Eq.~\eqref{eq:GausApprox}. The information rates are then obtained via Eq.~\eqref{I_receiver}.

\section{MAC HSW theorem}
\label{supp:note2}

\label{winter_theorem}
\begin{figure}[b]
    \centering
    \includegraphics[width=0.4\textwidth]{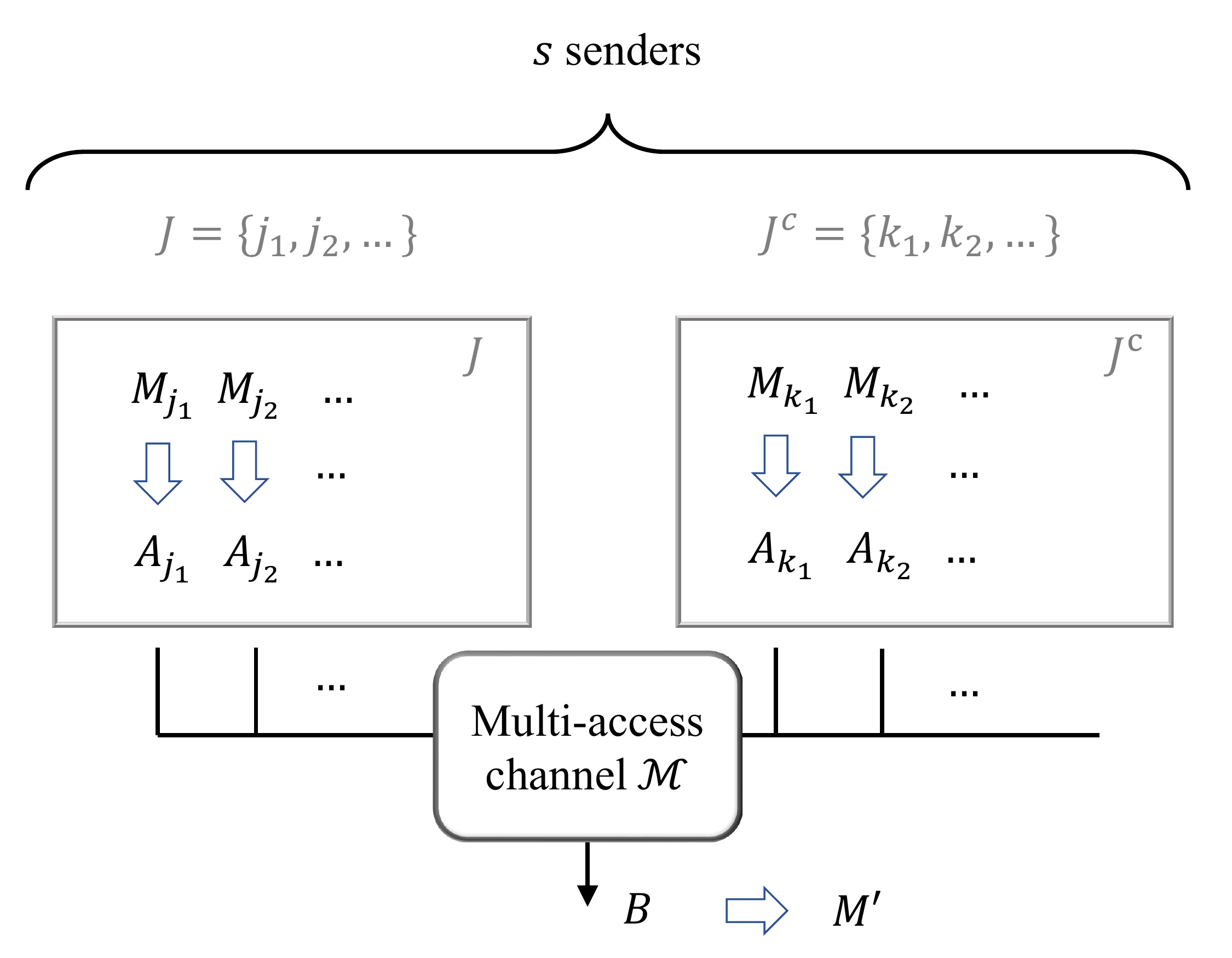}
    \caption{The protocol of general MAC communication. The $s$ senders individually prepare the quantum states in the corresponding quantum system $A_i$ given their messages in each codeword space $M_i$. The receiver reconstructs a message in codeword space ${M'}$ decoded from the output in quantum system ${B}$.}
    \label{fig:MACprotocol}
\end{figure}

Consider the $s$-sender MAC $\calM$. Here we deliberately use a different notation $\calM$ instead of $\calN$ for reasons that will become apparent later. The scenario of MAC classical communication without EA is similar to the EA case in Fig.~\ref{fig:EAMACschematic_main} of the main paper, except that the EA systems $A^\prime$ are gone. We also adopt the same notation, as shown in Fig.~\ref{fig:MACprotocol}.

The receiver decodes the messages $m=m_1m_2\ldots m_s$ of all senders, of which the word space is denoted as ${M'}$. Accordingly, we can write the message as a bipartition $m=m[J] m[J^c]$. The input encoding of the MAC $\{p_m,\hat{\sigma}^m_A\}$ involves sending each quantum state $\hat{\sigma}^m_A$ through MAC $\calM$ by probability $p_m$. The overall state after the encoding is
\be 
\hat{\Xi}_{MA}=\sum_m p_m \state{m}_M \otimes \hat{\sigma}^m_{A}
\,.\ee
Note that due to the independence between the $s$ senders, the overall probability and state of each encoding can be written as the products
$ 
p_m=\prod_{i=1}^s p_{m_i}\,, \hat{\sigma}^m_A=\otimes_{i=1}^s \hat{\sigma}_{A_i}\,,
$
where $A_i$ denote the quantum system of sender $i$.
Alternatively, for an arbitrary bipartition $J,J^c$, we have
$ 
p_{m}=p_{{m[J]}}p_{{m[J^c]}}\,, \hat{\sigma}^m_A=\hat{\sigma}_{{A}[J]}\otimes \hat{\sigma}_{{A}[J^c]}.
$ 
The MAC maps the composite input ${A}$ to a single output quantum system ${B}$, leading to the output 
\be 
\hat{\omega}_{MB}=\sum_m p_m \state{m}_M \otimes \calM_{A\to B} (\hat{\sigma}^m_A).
\label{omega:C}
\ee

With the above notations prepared, we can introduce the information theoretical quantities. To begin with, we introduce the conditional entropy
\be 
S\left({B}|M[J^c]\right)_{\hat{\omega}}=\sum_{{m[J^c]}} p_{{m[J^c]}} S\left[\calM_{A\to B}\left(\sum_{{m[J]}}p_{{m[J]}} \hat{\sigma}^m_A\right)\right],
\ee 
similarly, noting that $M=M[J]M[J^c]$ and $p_m=p_{{m[J]}}p_{{m[J^c]}}$ we can have the notation 
\be 
S\left({B}|M\right)_{\hat{\omega}}=\sum_{m}p_m S\left[\calM_{A\to B}\left(\hat{\sigma}^m_A\right)\right].
\ee 
We introduce the conditional quantum mutual information 
\bal
I\left(M[J]; B| M[J^c]\right)_{\hat{w}}=
 S\left({B}|M[J^c]\right)_{\hat{\omega}}-S\left({B}|M\right)_{\hat{\omega}}
\eal
 
\QZ{The one-shot capacity region of the quantum MAC $\calM$ can be directly obtained from the classical-quantum channel formalism in Ref.~\cite{winter2001capacity}. 
In this case, we can reduce the $(n,R_1,\ldots,R_s,\epsilon)$ code defined in Appendix~\ref{app:method} to the $n$-block code proposed in Ref.~\cite{winter2001capacity}, as we detail below. In the $n$ quantum channel uses, sender $k$ encodes the $n$-partite codeword $\bm w_k^n$ from the code space $\bm W_k^n$ of size $2^{nR_k}$. Without loss of generality, we consider each codeword to be chosen with equal probability $ p_{{\bm w}_k^n}=1/2^{nR_k}$. A classical $n$-block coding maps the codewords from $\bm W_k^n$ to $\bm M_k^n$ as
\be 
F:\bm w_k^n\to \bm m_k^n\equiv m_k^{(1)}\ldots m_k^{(n)}\,, 1\le k\le s,
\label{eq:ccencoding}
\ee 
where each block $m_k^{(u)}$ is in the message space $M_k^{(u)}$, the $n$ blocks together form message $\bm m_k^n$ in $\bm M_k^n$. Alternatively, $F$ can be implemented by sending each message $\bm m_k^n$ subject to probability distribution
\be 
p_{\bm m_k^n}=\sum_{\bm w_k^n\in \calS(\bm m_k^n)}p_{\bm w_k^n}\,,
\ee
where $\calS(\bm m_k^n)$ is the set spanned by $\bm w_k^n$'s satisfying $F(\bm w_k^n)=\bm m_k^n$. We define the marginal distribution for each block $p_{m_k^{(u)}}$ for $1\le u \le n$. 
For the one-shot capacity region, entanglement is forbidden among the $n$ blocks and we consider product states as the encoding \be 
\hat{\sigma}^{\bm m_k^n}=\otimes_{u=1}^n \hat{\sigma}^{m_k^{(u)}}_u\,,\, 1\le k\le s\,,
\ee 
where $\hat{\sigma}^{m_k^{(u)}}_u$ is the state of the quantum system at the $u$th block for sender $k$. Without loss of generality, we can always fix the classical-quantum encoding mapping
\be 
\calF_k:m_k^{(u)}\to \hat{\sigma}^{m_k^{(u)}}\,,\,1\le k\le s\,,
\label{eq:cqencoding}
\ee 
for all channel uses $1\le u \le n$. This is because time sharing of different $\{\hat\sigma^{m_k^{(u)}}_u\}$'s over $u\in [1, n]$ is available by expanding the domain of $\{m_k^{(u)}\}$ and varying the support of $\{p_{m_k^{(u)}}\}$ per block.
As a result, we can always combine the encoding $\{p_{{\bm m}_k^n},\hat{\sigma}^{{\bm m}_k^n}\}$ ($1\le k\le s$) with the quantum channel $\calN^{\otimes n}$ into an $n$-block classical-quantum channel ${\cal W}^{\otimes n}$, with each classical-quantum channel defined by
\be 
{\cal W}\equiv\calM\circ(\otimes_{k=1}^s\calF_k): m_1  m_2\ldots  m_s\to\calM(\otimes_{k=1}^s\hat{\sigma}^{m_k}).
\label{eq:cqchannel}
\ee 
To summarize, when considering the one-shot capacity region, a quantum MAC $\calM$ with any encoding $\{p_{\bm m_k^n},\hat{\sigma}^{\bm m_k^n}\}$, $1\le k\le s$, can be reduced to a classical-quantum channel ${\cal W}$ with certain classical coding $F$.
}

With the notations introduced, \QZ{we restate Theorem 9 and Theorem 10 of Ref.~\cite{winter2001capacity} in a version adapted to the communication over a quantum MAC without EA,} in preparation for the proof of Theorem~\ref{theorem: EA_MAC_main} of the main paper. 

\QZ{
\begin{lemma}
\label{lemma:mac_holevo}
Let $(R_1,\ldots, R_s)$ be any non-negative-valued vector, satisfying the constraints
\bal
\sum_{i\in J}R_i&\leq 
I\left(M[J]; B| M[J^c]\right)_{\hat{\omega}},
\label{Holevo_mutual}
\eal
for some encoding $\{p_{m_k},\hat{\sigma}^{m_k}\}$, $1\le k\le s$, where $\hat{\omega}$ is defined in Eq.~\eqref{omega:C} relying on the encoding. Then, for every $\epsilon,\delta>0$ and all sufficiently large $n$, there exists a
$(n,R_1-\delta,\ldots,R_s-\delta,\epsilon)$-code, i.e. $(R_1,\ldots, R_s)$ is achievable.
\end{lemma}
\begin{proof}
With the reduction to classical-quantum channel defined by Eq.~\eqref{eq:cqchannel} and Eq.~\eqref{eq:cqencoding}, this lemma immediately follows from Theorem 9 in Ref.~\cite{winter2001capacity}. We note that the capacity achieving encoding is subject to a product prior distribution $p_{\bm m_k^n}=\prod_{u=1}^n p_{m_k^{(u)}}$, and a product-state encoding $\hat{\sigma}^{\bm m_k^n}=\otimes_{u=1}^n\hat{\sigma}^{m_k^{(u)}}$.
\end{proof}
}

\begin{theorem}
\label{theorem:mac_hsw}
 For any $s$-sender quantum MAC $\calM$, \QZ{given a constrained set of classical-quantum encoding mappings $\calF_k: m_k\to \hat{\sigma}^{m_k}$, $1\le k\le s$, the one-shot constrained capacity} region of reliable classical communication over $\calM$ is \QZ{the convex hull over all rates $(R_1,R_2, \ldots, R_s)$ satisfying the $2^s$ inequalities}
\bal
\sum_{i\in J}R_i&\leq 
I\left(M[J]; B| M[J^c]\right)_{\hat{\omega}},
\label{HSW_mutual}
\eal
\QZ{for some encoding $\{p_{m_k}, \hat{\sigma}^{m_k}\}$, $1\le k\le s$.} Here $\hat{\omega}$ is defined in Eq.~\eqref{omega:C} relying on the encoding.
\end{theorem}
\begin{proof}
\QZ{
The achievability results from Lemma \ref{lemma:mac_holevo}, which states that for every $\epsilon, \delta>0$ and all sufficiently large $n$, there exists a
$(n,R_1-\delta,\ldots,R_s-\delta,\epsilon)$ code (see Appendix~\ref{app:method}) achieving $I\left(M[J]; B| M[J^c]\right)_{\hat{\omega}}$. Thus, by time-sharing, Eq.~\eqref{HSW_mutual} can always be achieved. Now we address the proof of weak converse here. Considering the classical coding $F$ in, let $\left\{\{\hat{\sigma}_{{\bm w}_k^n}\}\,|\,1\le k\le s\right\}$, $\{\hat{\Lambda}_{\bm w_1^n\cdots \bm w_s^n}\}$ be any $(n,R_1, \ldots, R_s, \epsilon)$ code. The channel output is in a separable state
\be 
\hat \omega=\otimes_{k=1}^s\frac{1}{2^{nR_k}}\sum_{{\bm w}_k^n}({\cal W}^{\otimes n}\circ F)[{\bm w}_k^n]\otimes \ketbra{{\bm w}_k^n}{{\bm w}_k^n}.
\ee
Its reduced state in the $u$th channel use is denoted as $\hat\omega^{(u)}$.
Denote the error rate of the POVM output $\bm w_k^{\prime n}$ as $\epsilon=P({\bm w}_k^n\neq \bm w_k^{\prime n})$. By Fano's inequality we have 
\be 
S({\bm W}[J]^{n}|{\bm W}^{\prime n})_{\hat\omega}\le 1+\epsilon\, n\sum_{i\in J}R_i\,.
\ee
Note that $|\bm w_k^n|=2^{nR_k}, 1\le k \le s$. We have
\bal 
n\sum_{i\in J}R_i&=S(\bm W[J]^{n})_{\hat\omega}\\
&=I({\bm W}[J]^{n};{\bm W}^{\prime n})_{\hat\omega}+S({\bm W}[J]^{n}|{\bm W}^{\prime n})_{\hat\omega}\\
&\le I({\bm W}[J]^{n};{\bm W}^{\prime n})_{\hat\omega}+1+\epsilon\, n\sum_{i\in J}R_i\,.
\eal 
For reliable communication we require $\epsilon\to 0$, thus
\bal 
\sum_{i\in J}R_i&\le  \frac{I({\bm W}[J]^{n};{\bm W}^{\prime n})_{\hat\omega}}{n}+\frac{1}{n}\\
&\le  \frac{I({\bm M}[J]^{n};{\bm B}^{n}|{\bm M}[J^c]^{n})_{\hat\omega}}{n}+\frac{1}{n}\\
&\le  \frac{\sum_{u=1}^n I(M[J]^{(u)};B^{(u)}|M[J^c]^{(u)})_{\hat\omega^{(u)}}}{n}+\frac{1}{n}\,,
\eal 
which is outer bounded by the convex hull of rates defined by Eq.~\eqref{HSW_mutual}. The second inequality is due to the data processing inequality since we have the completely positive trace-preserving mapping $\bm W^n\to \bm M^n\to \bm B^n\to \bm W^{\prime n}$, and that discarding system never increases quantum mutual information (which follows from strong subadditivity of von Neumann entropy); the third inequality results from the chain rule and the fact that $\hat{\omega}$ conditioned on $\bm M^n$ is a product state over $n$ channel uses. Specifically, choosing the encoding over $n$ blocks to be identical, we have the outer bound of one-encoding rate region defined by
\bal 
\sum_{i\in J}R_i&
&\le I(M[J];B|M[J^c])_{\hat\omega}+\frac{1}{n}\,,
\label{eq:weak converse}
\eal 
}
\end{proof}

\QZ{
\textbf{Remark.} 
For an $\ell$-letter communication scenario (encoding with inputs entangled between $\ell$ channel uses), the systems ${\bm A}^{\ell}$ and ${\bm B}^{\ell}$ can be entangled over the $\ell$ letters within a block. In this case, the mutual information is not necessarily subadditive, Eq.~\eqref{HSW_mutual} needs to be evaluated over $\ell$ channel uses as a whole
\be 
\sum_{i\in J}R_i\leq  I\left({\bm M}[J]^{\ell}; {\bm B}^{\ell}| {\bm M}[J^c]^{\ell}\right)_{\hat{\omega}}\,.
\ee
This has already been pointed out in Sec. VII of Ref.~\cite{winter2001capacity}. 
}

In the EA protocol, the quantum system comes with a reference system $A^\prime$ which is pre-shared to the receiver side intact, while the quantum system goes through the MAC $\calN$. However, we can view the overall channel $\calM= \calN\otimes\calI$ as a single MAC. Note that the encoding is, however, restricted as the $\calI$ part corresponds to references. \QZ{Namely, the encoding mappings are defined by
\be 
\calF_k:m_k\to \calE_{{ m}_k}(\hat{\sigma}^0),\,1\le k\le s\,,
\label{eq:EAencoding}
\ee
where the generator state $\hat{\sigma}^0$ is an arbitrary state of $A_k A_k^{\prime}$, while the encoding operations $\{\calE_{{ m}_k}\}$ is constrained to act only on the signal system $A_k$.}

\section{Proof of Theorem~\ref{theorem: EA_MAC_main} of the main paper}
\label{supp:note3}

We derive the EA-MAC classical capacity in analogy to the single-sender single-receiver EA classical capacity~\cite{bennett2002entanglement} and the two-sender EA-MAC classical capacity~\cite{hsieh2008entanglement}. 

\QZ{Theorem \ref{theorem: EA_MAC_main} in the maintext includes two parts: the regularization Eq.~\eqref{eq:CE_regularization}, and the convex hull Eq.~\eqref{eq:CE_union_state}.
}

\QZ{To begin with, we provide an explanation for the regularization. For channel $\calN$, denote ${\cal X}(\ell)$ as the set of all codebooks with at most $\ell$ entangled letters. The $\ell$-letter one-shot capacity region $\calC_{\rm E}^{(1)}(\calN^{\otimes \ell})$ spans the rate subset achieved by the codebook set ${\cal X}(\ell)$. Thus the union of block codes $\bigcup_{\ell=1}^\infty{\cal X}(\ell)$ includes any possible codebook. Hence the ultimate capacity region $\calC_{\rm E}(\calN)$ is given by the regularized rate region achieved by the universal codebook set $\bigcup_{\ell=1}^\infty{\cal X}(\ell)$. This block-code reduction yields Eq.~\eqref{eq:CE_regularization} from Eq.~\eqref{eq:CE_union_state} in the main text.}

\QZ{Below we provide both the inner and outer bounds of single-letter one-shot region $\calC_{\rm E}^{(1)}(\calN)$, which coincides and thus gives the formula Eq.\eqref{eq:CE_union_state} in the main text. The proof does not specify the mode number per user since the entanglement assistance is unlimited, thus the formula of $\ell$-letter one-shot region $\calC_{\rm E}^{(1)}(\calN^{\otimes \ell})$ follows by inserting $\calN\to\calN^{\otimes \ell}$, $A\to \bm A^{ \ell}$, $B\to \bm B^{ \ell}$, $A'_k$ purifying $\bm A_k^{ \ell}$, $1\le k \le s$. The $\ell$-letter one-shot, one-state region $\tilde\calC_{\rm E}(\calN^{\otimes \ell},\hat{\phi})$ is defined by
\be 
\sum_{k\in J}R_k\leq I(A'[J];\bm B^{ \ell}|A'[J^c])_{\hat{\rho}}, \forall J\,,
\ee
with $\hat{\rho}=\calN^{\otimes \ell}\otimes \calI(\hat{\phi})$, where $\hat{\phi}$ is a pure state defined on $\bm A^{\ell}A'$. } 

In the proof below, we omit ` $\hat{}$ ' in the notation of operators for simplicity, since the meaning of each notation is clear in the context.
\QZ{
Our proof combines the techniques in Ref.~\cite{winter2001capacity} (Lemma~\ref{lemma:mac_holevo} and Theorem~\ref{theorem:mac_hsw}), Ref.~\cite{bennett2002entanglement}, and Ref.~\cite{hsieh2008entanglement}.
}

\subsection{Inner bound}

We prove the inner bound by showing the boundaries achievable for any pure state $\phi$. We show that the $2^s$ conditional quantum mutual information quantities in Eq.~\eqref{HSW_mutual} reach all the boundaries of $\tilde \calC_{\rm E}$ region defined by Eq.~\eqref{eq:CeMAC_qinfo} of the main paper, \QZ{i.e. the rate region defined by Eq.~\eqref{HSW_mutual} is inner bounded by $\tilde \calC_{\rm E}$ regions.} \QZ{Lemma~\ref{lemma:mac_holevo}} guarantees that \QZ{for any $\epsilon,\delta>0$ there is some $(n,R_1-\delta,\ldots,R_s-\delta,\epsilon)$ code achieving the region defined by} Eq.~\eqref{HSW_mutual} \QZ{with some encoding $\{p_{ m_k},\sigma^{m_k}\}, 1\le k\le s$, thus the convex hull over all possible encodings} provides the inner bound \QZ{of $\calC_{\rm E}^{(1)}(\calN)$}. Then the achievability \QZ{of $\text{Conv}\left[\bigcup_{\phi}\tilde\calC_{\rm E}(\calN,\phi)\right]$} is proven. During the evaluation of entropies within this section, $AA'$ \QZ{or $BA'$} is the relevant Hilbert space and we omit the subscripts for states living in $AA'$ \QZ{or $BA'$}.

We begin with the special case of each pair $\{({A}_i,{A^\prime_i}), 1\leq i\leq s\}$ being in a maximally entangled state, such that the reduced state $\phi_{{A}_i}=\Tr_{{A^\prime}{A}_{j\neq i}}(\phi_{})=(I/d_i)_{{A}_i}$ is fully mixed, here the dimension $d_i=\text{dim}{A}_i$. The capacity achieving protocol is implemented by a generator state being $\phi$ and a correlation-removing unitary encoding $\{U_{m}\}$. In terms of the ensemble $\sum_{m} p_{m} U_{m}\phi U_{m}^\dagger$, a correlation-removing encoding $\{U_{m}\}$ wipes out the correlation between each system ${A}_i$ with its reference system. Concretely, this can be implemented by the generalized Pauli operators 
\bal
U_{m_i}&=T_{{A}_i}^{m_{i,T}}R_{{A}_i}^{m_{i,R}}, \,\text{where}\\ 
m_i&=m_{i,T}m_{i,R}, 1\leq m_{i,T}\leq d_i, 1\leq m_{i,R}\leq d_i
\,,\eal
with the Pauli matrices $T_{jk}=\delta_{j,k-1 \text{ mod }d_i}$, $R_{jk}=e^{i2\pi k/d_i}\delta_{jk}$ acting on each subspace ${A}_i$. Governed by the independence constraint of MAC, $\phi_{}=\otimes_{i=1}^s\phi_{{A}_i{A^\prime_i}}$. The encoding $U_{m}=\otimes_{i=1}^s U_{m_{i}}$ acting on $\otimes_{i=1}^s{A}_i$ by probability $p_{m}=\prod_{i=1}^sp_{m_{i}}$ is also separable, \QZ{where $m_i$ is subject to the uniform distribution $p_{m_i}=1/|M_i|$}.

Now the overall quantum state after the channel is
\be 
\omega_{MBA^\prime}=\sum_m p_m \state{m}_M \otimes \left[\calN_{A\to B}  \otimes \calI\right] \left(U_{m}\phi_{}U_{m}^\dagger\right).
\label{eq:omega_app}
\ee 
Applying Lemma~\ref{lemma:mac_holevo} to Eq.~\eqref{eq:omega_app}, \QZ{for any $\epsilon,\delta>0$ there exists a $(n,R_1-\delta,\ldots,R_s-\delta,\epsilon)$ code achieving the vector of data transmission} rate of this protocol 
\begin{widetext}
\bal
\sum_{i\in J}R_i&= I\left(M[J]; B| M[J^c]\right)_{\omega}
\\
&=\mathbb{E}_{{m[J^c]}} \Bigg\{S\left[\calN\otimes\calI \left(\sum_{{m[J]}} \QZ{p_{m[J]}\left(U_{{m[J]}}\otimes U_{{m[J^c]}}\right) \phi_{} \left(U_{{m[J^c]}}^\dagger \otimes U_{{m[J]}}^\dagger\right)}\right)\right]-\mathbb{E}_{{m[J]}} \Big\{S\left[\calN\otimes\calI\left(U_{m}\phi_{}U_{m}^\dagger\right)\right]\Big\}\Bigg\},
\\
&=\mathbb{E}_{{m[J^c]}} \Bigg\{S\left[\calN\otimes\calI\left(U_{{m[J^c]}}\phi_{{{A^\prime}[J^c]}{A}}U_{{m[J^c]}}^\dagger\right)\right]+S\left(\phi_{{A^\prime}[J]}\right) -\mathbb{E}_{{m[J]}} \Big\{S\left[\calN\otimes\calI\left(U_{m}\phi_{}U_{m}^\dagger\right)\right]\Big\}\Bigg\},
\label{eq:chi_cond}
\eal
where $\mathbb{E}_{x}[f(x)]\equiv \int dx p(x) f(x)$ refers to the expectation value of function $f$ averaged on probability distribution $p$. The second equality is because $\{U_{{m[J]}}\}$ is correlation-removing within system $A[J]$,
\bal
\calN\otimes\calI \left(\sum_{{m[J]}} \QZ{p_{m[J]}\left(U_{{m[J]}}\otimes U_{{m[J^c]}}\right) \phi_{} \left(U_{{m[J^c]}}^\dagger \otimes U_{{m[J]}}^\dagger\right)}\right)
=\calN\otimes\calI \left(U_{{m[J^c]}}\phi_{{A^\prime[J^c]}{A}}U_{{m[J^c]}}^\dagger\right)\otimes \phi_{{A^\prime}[J]}
\,.\eal
\end{widetext}
In contrast to the single-user case, here we see that the encoding restricted on ${A}[J]$ fails to remove the correlation between ${A}$ and ${{A^\prime}[J^c]}$, so ${{A^\prime}[J^c]}$ is left with ${A}$ in the entropy quantities. 

Next, we note that for the maximally entangled state $\phi$, we have $U\otimes I \ket{\phi}=I\otimes U^\star\ket{\phi}$ for any unitary $U$.
This equivalence ensures that the encoding preserves the entropy of the output state
\begin{align} 
&S\left[\calN\otimes\calI\left(U_{{m[J^c]}}\phi_{{{A^\prime}[J^c]}{A}}U_{{m[J^c]}}^\dagger\right)\right]
\nonumber
\\
&\qquad =S\left[\calN\otimes\calI\left(\phi_{{{A^\prime}[J^c]}{A}}\right)\right],
\end{align}
similarly 
\be 
S\left[\calN\otimes\calI\left(U_{m}\phi_{}U_{m}^\dagger\right)\right]=S\left[\calN\otimes\calI\left(\phi_{}\right)\right].\ee 
Therefore, the averaging in $\mathbb{E}_{{m[J]}}, \mathbb{E}_{{m[J^c]}}$ disappears in Eq.~\eqref{eq:chi_cond}. 
Therefore, Eq.~\eqref{eq:chi_cond} reduces to
\begin{align}
\sum_{i\in J}R_i=&S\left(\calN\otimes \calI \left({\phi}_{A'[J^c]A}\right)\right)+S(\phi_{{A^\prime}[J]})
\nonumber
\\
&-S\left(\calN\otimes \calI \left(\phi\right)\right).
\label{eq:CeMAC_old} 
\end{align}
Simplifying the right-hand-side above, we obtain
\bal
\sum_{i\in J}R_i&=  S\left(BA'[J^c]\right)_{{\rho}}+S\left(A'[J]\right)_{{\rho}}-S\left(A'[J]BA'[J^c]\right)_{{\rho}}
\\
&=I\left(A'[J];BA'[J^c]\right)_{{\rho}}=I(A'[J];B|A'[J^c])_{{\rho}}.
\label{Eq_Rsum}
\eal
Here we utilized the channel mapping
$ 
\QZ{\rho=}{\rho}_{BA^\prime}=\calN_{A\to B}\otimes \calI ({\phi}_{AA^\prime})
$.
\QZ{The} last step is due to independence between the signals.
Note that Eq.~\eqref{Eq_Rsum} holds for any $J$. Hence Eq.~\eqref{eq:chi_cond} achieves $\tilde \calC_{\rm E}(\calN,{\phi})$ in Eq.~\eqref{eq:CeMAC_qinfo} of the main paper for the maximally entangled state.

Indeed, the achievability also holds for any \emph{projection-like} $\phi_A$, which is proportional to a projection operator to some subspace $\calS$ of the input Hilbert space $\calH_{\rm in}$, along with the encoding $\{U_{{m[J]}}\}$ restricted on $\calS$. 
For now, we have achieved the rate region specified by Eq.~\eqref{eq:CeMAC_qinfo} of the main paper for all projection-like inputs.

Having proven the result for the special case of inputs with a projection-like reduced density matrix, now we extend the proof to general pure product states $\phi$. For simplicity here we denote the reduced density matrix as $\xi\equiv\phi_A$.
{
First we introduce the $n$-\QZ{block} $\epsilon$-typical space $T_n(\xi)$ of an arbitrary single-\QZ{block} state $\xi=\sum_x p(x)\ketbra{x}{x}$, where $\{\ket{x_i}\}$ is an orthogonal set. $T_n(\xi)$ comprises every typical states $\ket{x_1x_2\ldots x_n}$ associated with typical sequences $x_1x_2\ldots x_n$ satisfying
\be
\left|-\log(\prod_{i=1}^n p(x_i))/n-S(\xi)\right|\leq\epsilon
\label{eq:typicality}
\,.\ee 
Denote by $P_{T_n}$ the projection operator onto subspace $T_n(\xi)$. The typicality follows from the law of large number, explicitly, for any $\epsilon>0$ there exists $n$ sufficiently large such that the typical subspace almost includes the $n$-\QZ{block} random encoding $\xi^{\otimes n}$
\be
\Tr \left[\xi^{\otimes n}\left(I-P_{T_n}\right)\right]\leq \delta.
\ee 
}
Now let $\pi_{T_n}$ be the normalized projection operator $P_{T_n}/\dim T_n$. By the definition of $\epsilon$-typicality Eq.~\eqref{eq:typicality}, the entropy of the uniform distribution $\pi_{T_n}$ is almost $nS(\xi)$ up to a prefactor $\sim 1-\delta$, explicitly 
\be 
(1-\delta)2^{n(S(\xi)-\epsilon)}\leq \dim T_n(\xi)\leq 2^{n(S(\xi)+\epsilon)}. 
\label{eq:pi_typicality}
\ee
By a similar correlation-removing encoding on \QZ{$\Phi_{\pi_{T_n}}$}, below we show that the $n$-\QZ{block} $I(\bm A'[J]^n;\bm B^n|\bm A'[J^c]^n)_{(\calN\otimes \calI)^{\otimes n} (\pi_{T_n})}$ quantity achieves $nI(A'[J];B|A'[J^c])_{\calN\otimes \calI (\xi)}$.

\QZ{Following a similar procedure to the proof for maximally entangled state, Eq.~\eqref{eq:CeMAC_old} holds for the $n$-block state}. The typicality Eq.~\eqref{eq:pi_typicality} gives the first term in Eq.~\eqref{eq:CeMAC_old}. Meanwhile, for the last two terms in Eq.~\eqref{eq:CeMAC_old}, we construct a unitary including the environment mode to simulate channel $\calN$. Denote the complementary channel as $\calK$, which maps the input state to the environment mode. Define $\Phi_{\pi_{T_n}}$ as a purification of $\pi_{T_n}$ fulfilled per system pair $A_iA'_i$, $1\le i\le s$. Then the last term $S[(\calN\otimes\calI)^{\otimes n}(\Phi_{\pi_{T_n}})]=S[\calK(\pi_{T_n})]$, thereby the latter two terms converges to the desired quantity when $\epsilon\to 0$
\bal
&S\left[\calN^{\otimes n}\left(\pi_{T_n}\right)\right]-S\left[\calK\left(\pi_{T_n}\right)\right]
\\
&\to n\left[S\left(\calN\left(\xi\right)\right)-S\left(\calN\otimes\calI\left(\Phi_{\xi}\right)\right)\right]
\,,\eal
which follows from Lemma 1 in \cite{bennett2002entanglement}. Both sides of the equation above can be reduced to a conditional quantum information quantity: $I(\bm A'[J]^{ n};\bm B^{ n}|\bm A'[J^c]^{ n})_{{\rho}}$ with $\rho=(\calN\otimes \calI)^{\otimes n}(\Phi_{\pi_{T_n}})$ for LHS, and $nI(A'[J];B|A'[J^c])_{{\rho}}$ with $\rho=\calN\otimes \calI(\phi)$ for RHS. Combining the above together, we have the information rate per letter
\bal
I(\bm A'[J]^{ n};&\bm B^n|\bm A'[J^c]^n)_{(\calN\otimes \calI)^{\otimes n} (\Phi_{\pi_{T_n}})}/n\\
&\to I(A'[J];B|A'[J^c])_{\calN\otimes \calI (\phi)}
\eal
given $\epsilon\to 0$. Note that $\pi_{T_n}$ is projection-like thus $I(\bm A'[J]^{n};\bm B^{n}|\bm A'[J^c]^{n})_{(\calN\otimes \calI)^{\otimes n} (\Phi_{\pi_{T_n}})}$ is \QZ{achieved by the correlation-removing encoding}. Hence, $I(A'[J];B|A'[J^c])_{\calN\otimes \calI (\phi)}$ is achievable for any pure product state $\phi$.  \\

Now we have proven $\tilde \calC_{\rm E}$ specified by Eq.~\eqref{eq:CeMAC_qinfo} of the main paper is achievable for any input $\phi$. \QZ{Then the convex hull of $\tilde \calC_{\rm E}$ is achievable due to time sharing.} Thus the one-shot capacity region is inner bounded by the \QZ{convex hull} of $\tilde \calC_{\rm E}(\calN,\phi_{})$ regions over all possible $\phi_{}$
\be
\calC_{\rm E}^{(1)}(\calN)\supseteq \QZ{\rm Conv}\left[\bigcup_{\phi_{}}\tilde \calC_{\rm E}(\calN,\phi_{})\right]\,.
\ee

\subsection{Outer bound}
\begin{figure}[htbp]
    \centering
    \includegraphics[width=0.25\textwidth]{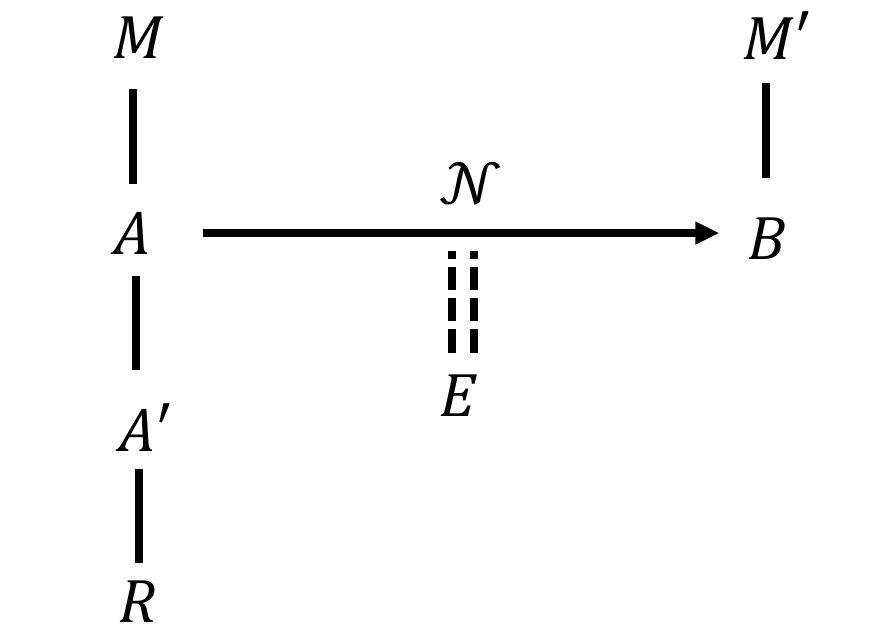}
    \caption{For the proof of the outer bound, $A,B$ are the input and output systems of our MAC $\calN$, $M,M'$ are the input and output code spaces. $A'$ is the entanglement assistance system. System $R$ purifies system $A$ per subsystem $A_i$, $1\le i \le s$.}
    \label{fig:lemma}
\end{figure}

\QZ{Theorem~\ref{theorem:mac_hsw} gives an outer bound with a convex hull over the region defined by $I(M[J];BA'|M[J^c])_{\omega}$ quantities in Eq.~\eqref{HSW_mutual}. We prove that $\tilde C_{\rm E}(\calN,\phi)$ is an outer bound of each region defined by Eq.~\eqref{HSW_mutual} with repetitive uses of one encoding $\{p_m\}$, which is the outer bound of one-shot, one-encoding rates. Then we take the convex hull over $\tilde C_{\rm E}(\calN,\phi)$ and prove it as an outer bound of $\calC_{\rm E}^{(1)}(\calN)$. }

To ease the proof, we introduce reference systems $R, E$, as shown in Fig.~\ref{fig:lemma}. $R=\otimes_{i=1}^s R_i$ individually purifies per sender $i$ the subsystem $\Xi_{M_iA_iA'_i}$ of the overall state $\Xi=\Xi_{MAA'}$ just after encoding, defined by Eq.~\eqref{eq:Xi} of the main paper. The environment system $E$ purifies the overall state,  including the system $R$, just after the channel. The channel is extended to a unitary transform $\calU_\calN:{AE\to BE}$ when $E$ is included. We denote the purification before channel as 
\be 
\xi=\xi_{MAA'RE}\QZ{=\Phi_{\Xi}\otimes \ketbra{0}{0}_E\,,}
\label{eq:xidef}
\ee
and the final purification after channel as $\eta=\eta_{MBA'RE}=\calU_\calN(\xi)$. Below we show that $I(M[J];BA'|M[J^c])_{\omega}$ is upper bounded by the right-hand side of Eq.~\eqref{eq:CeMAC_qinfo} of the main paper.

Equivalently, since the purification systems $R,E$ are traced out at the end and make no difference, we evaluate $I(M[J];BA'|M[J^c])_{\omega}$ over $\eta$
\bal
&I\left(M[J];BA'\big|M[J^c]\right)_{\eta}
=I\left(M[J];BA'M[J^c]\right)_{\eta}
\\
&=I\left(M[J]A'[J];BA'[J^c]M[J^c]\right)_{\eta}
\\
&\qquad-I\left(A'[J];BA'[J^c]M[J^c]\right)_{\eta}+I\left(M[J];A'[J]\right)_{\eta}
\,.\eal
The first equality is due to \QZ{the chain rule of quantum mutual information and the independency} constraint of MAC $I(M[J^c];M[J])=0$. The second equality is due to the chain rule of quantum mutual information. Let $M[J]=X$, $BA'[J^c]M[J^c]=Y$ and $A'[J]=Z$, we have $I(M[J];BA'M[J^c])_{\eta}=I(X;YZ)=I(Y;X|Z)+I(X;Z)=I(XZ;Y)-I(Z;Y)+I(X;Z)$, which equals the second line term by term. Note that the encoding does not affect the entanglement assistance $A'$, thus
\be  
I(M[J];A'[J])_{\eta}=0\,,
\ee
and from the positivity of quantum mutual information
\be 
I(A'[J];BA'[J^c]M[J^c])_{\eta}\geq 0
\,,
\ee 
we have
\be
I(M[J];BA'|M[J^c])_{\eta}\leq I(M[J]A'[J];BA'[J^c]M[J^c])_{\eta}
\,.
\label{App:I2}
\ee
As EA is unlimited, there is always an expanded system $A^{\prime}_{\rm ex}=\otimes_{i=1}^s A^{\prime}_{{\rm ex},i}$ available that includes $A'_i, R_i,M_i$ for each sender $i$. Adopting the expanded EA,
\be
 I(M[J]A'[J];BA'[J^c]M[J^c])_{\eta}\leq  I(A^{\prime}_{\rm ex}[J];BA^{\prime}_{\rm ex}[J^c])_{\eta}
\,,\ee
since, \QZ{as a result of strong subadditivity of the Von Neumann entropy,} discarding systems never increases quantum information.
Note that $I(A^{\prime}_{\rm ex}[J^c];A^{\prime}_{\rm ex}[J])=0$, combining Eq.~\eqref{App:I2} and the equation above, we obtain the outer bound 
\be
\QZ{I(M[J];BA'|M[J^c])_{\eta}}\leq I(B;A^{\prime}_{\rm ex}[J]|A^{\prime}_{\rm ex}[J^c])_{\eta}
\ee
for some pure state $\eta$, which is a product between different senders.
With the expanded entanglement assistance $A^\prime\equiv A^{\prime}_{\rm ex}$, we arrive at the desired outer bound for any encoding $\Xi$
\be 
\QZ{I(M[J];BA'|M[J^c])_{\eta}}\le I(B;A'[J]|A'[J^c])_\rho
\label{eq:EA_outer_bound_app}
\,,\ee
evaluated on the channel output $\rho=\calN\otimes\calI(\phi)$ of some pure product state $\phi=\phi_{AA'}$ that purifies the reduced state $\Xi_{A}$.

Denote $\calT_{\phi}$ as a family of encodings, of which the overall state after the encoding satisfies $\Phi_{\Xi_A}=\phi$. \QZ{Substitute the repetitive $n$ uses of some encoding from $\calT_{\phi}$ for the input of $\{\omega^{(u)}\}_{u=1}^n$ in Ineq.~\eqref{eq:weak converse}, the left hand sides of Ineqs.~\eqref{eq:EA_outer_bound_app} gives an outer bound for information rates of the repetitive encoding. Thus one-shot, one-encoding rates are outer bounded by the region $\tilde \calC_{\rm E}(\calN,\phi)$
\be 
\sum_{i\in J}\tilde R_i\le I(B;A'[J]|A'[J^c])_\rho
\,.\ee
Finally, taking the convex hull on both sides of Ineq.~\eqref{eq:EA_outer_bound_app} and combining it with Theorem~\ref{theorem:mac_hsw}, for any $(n,R_1,\ldots,R_s,\epsilon)$ code we have
\bal 
\sum_{i\in J}R_i&\le \sum_up_uI(M[J];BA'|M[J^c])_{\eta^{(u)}}\\
\eal
for some encoding set $\{p_{\bm m_k^{(u)}}, \sigma^{\bm m_k^{(u)}}\}$ that generates state set $\{\eta^{(u)}\}$ by plugging $\Xi=\sum_{\bm m_k^n}p_{\bm m_k^n}\sigma_{\bm m_k^n}\otimes\ketbra{\bm m_k^n}{\bm m_k^n}$ in Eq.~\eqref{eq:xidef}, and thereby
\bal 
\sum_{i\in J}R_i&\le \sum_up_uI(B;A'[J]|A'[J^c])_{\rho^{(u)}}\\
\eal
for some pure state sets $\{\phi^{(u)}=\Phi_{\Xi_{A^{(u)}}}\}$ that generate state sets $\{\rho^{(u)}\}$ by $\rho^{(u)}=\calN\otimes\calI(\phi^{(u)})$, with $\{p_u\}$ satisfying $\sum_up_u=1$. Thus the convex hull $\calC_{\rm E}^{(1)}(\calN)$ provides an outer bound for the $(n,R_1,\ldots,R_s,\epsilon)$ code rate}
\be
\calC_{\rm E}^{(1)}(\calN)\subseteq 
\QZ{\rm Conv}\left[\bigcup_{{\phi} } \tilde \calC_{\rm E}\left(\calN,{\phi} \right)\right]
\,.\ee

\end{appendix}

\end{document}